\documentclass[submission,copyright,creativecommons,sharealike,noncommercial]{eptcs}
\usepackage{breakurl}
\providecommand{\event}{QPL 2015}

\usepackage{graphicx}
\usepackage{amsfonts}
\usepackage{amsmath}  
\usepackage{amsthm}
\usepackage{amssymb}
\usepackage{relsize}
\usepackage{mathtools}
\usepackage[applemac]{inputenc} 
\usepackage[nocompress]{cite}

\usepackage[usenames,dvipsnames]{xcolor} 
\usepackage{tikz}
\usepackage{tikzfig}

\usetikzlibrary{arrows,shapes,decorations,backgrounds,positioning}

	\newcounter{theorem_c} 
	\numberwithin{theorem_c}{section} 
	\numberwithin{equation}{section} 

	\theoremstyle{plain} 
	\newtheorem{theorem}[theorem_c]{Theorem}
	
	\newtheorem{lemma}[theorem_c]{Lemma}

	\newtheorem*{theorem*}{Theorem}
	\newtheorem*{lemma*}{Lemma}
	\newtheorem*{corollary*}{Corollary}

	\theoremstyle{definition} 
	\newtheorem{definition}[theorem_c]{Definition}
	
	\newtheorem{remark}[theorem_c]{Remark}
	\newtheorem*{definition*}{Definition}
	\newtheorem*{example*}{Example}
	\newtheorem*{remark*}{Remark}

	\newcommand\numberthis{\addtocounter{equation}{1}\tag{\theequation}} 

	\newcommand{\reals}{\mathbb{R}} 
	\newcommand{\integersMod}[1]{\mathbb{Z}_{#1}} 
	\newcommand{\domain}[1]{\operatorname{dom}#1} 
	
	\newcommand{\suchthat}[2]{\left\{#1 \: \left\vert \: #2 \right.\right\}} 

		\newcommand{\ket}[1]{\vert #1 \rangle} 
		\newcommand{\bra}[1]{\langle #1 \vert} 

		\newcommand{\completeGraph}[1]{\operatorname{K}_{#1}} 
		\newcommand{\graphNode}[1]{\bullet{#1}}		
		\newcommand{\graphNodel}[1]{{#1}\bullet}
		\newcommand{\graphEdge}[2]{\graphNodel{#1}-\graphNode{#2}}
		\newcommand{\RelGraph}[1]{\mathcal{G}_{#1}}

		\newcommand{\SpaceH}{\mathcal{H}}

		\newcommand{\isom}{\cong} 
		\newcommand{\epim}{\twoheadrightarrow} 
		\newcommand{\monom}{\rightarrowtail} 
		\newcommand{\id}[1]{id_{#1}} 

		\newcommand{\tensor}{\otimes} 
		\newcommand{\bigtensor}{\bigotimes} 
		\newcommand{\tensorUnit}{I} 

		\newcommand{\SetCategory}{\operatorname{Set}} 
		\newcommand{\fdHilbCategory}{\operatorname{fdHilb}} 
		\newcommand{\AbCategory}{\operatorname{Ab}} 
		\newcommand{\fRelCategory}{\operatorname{fRel}} 

		\newcommand{\CategoryC}{\mathcal{C}}

		\newcommand{\CPMCategory}[1]{\operatorname{CPM}[#1]} 

	\newcommand{\Xcolour}{Red}
	\newcommand{\groupStructColour}{\Xcolour}

	\newcommand{\Zcolour}{YellowGreen}
	\newcommand{\classicalStructColour}{\Zcolour}
	\newcommand{\ZmultSym}{\hbox{\begin{tikzpicture} [scale=1,transform shape] 

\def\deltax{0.3} 
\def\deltay{0.5} 

\node (mult_label_inl) at (-\deltax,-\deltay) {};
\node (mult_label_inr) at (+\deltax,-\deltay) {};
\node [dot, fill=\classicalStructColour] (mult) at (0,0) {};
\node (mult_label_out) at (0,+\deltay) {};
\draw[-] [out=90,in=225](mult_label_inl) to (mult);
\draw[-] [out=90,in=315](mult_label_inr) to (mult);
\draw[-] (mult) to (mult_label_out);

\end{tikzpicture}}\!}
	\newcommand{\ZcomultSym}{\hbox{\begin{tikzpicture} [scale=1,transform shape] 

\def\deltax{0.3} 
\def\deltay{0.5} 

\node (mult_label_outl) at (-\deltax,+\deltay) {};
\node (mult_label_outr) at (+\deltax,+\deltay) {};
\node [dot, fill=\classicalStructColour] (mult) at (0,0) {};
\node (mult_label_in) at (0,-\deltay) {};
\draw[-] [in=270,out=135] (mult) to (mult_label_outl);
\draw[-] [in=270,out=45] (mult) to (mult_label_outr);
\draw[-] (mult_label_in) to (mult);

\end{tikzpicture}}\!}
	\newcommand{\ZunitSym}{\hbox{\begin{tikzpicture} [scale=1,transform shape] 

\def\deltax{0.3} 
\def\deltay{0.5} 

\node [dot, fill=\classicalStructColour] (mult) at (0,0) {};
\node (mult_label_out) at (0,+\deltay) {};
\draw[-] (mult) to (mult_label_out);

\end{tikzpicture}}\!}
	\newcommand{\ZcounitSym}{\hbox{\begin{tikzpicture} [scale=1,transform shape] 

\def\deltax{0.3} 
\def\deltay{0.5} 

\node [dot, fill=\classicalStructColour] (mult) at (0,0) {};
\node (mult_label_in) at (0,-\deltay) {};
\draw[-] (mult_label_in) to (mult);

\end{tikzpicture}}\!}
	\newcommand{\Zstructure}{\ZmultSym,\ZcomultSym,\ZunitSym,\ZcounitSym}
	
	\newcommand{\Zaltcolour}{Cyan}
	\newcommand{\internalclassicalStructColour}{\Zaltcolour}
	
	\newcommand{\ZaltcomultSym}{\hbox{\begin{tikzpicture} [scale=1,transform shape] 

\def\deltax{0.3} 
\def\deltay{0.5} 

\node (mult_label_outl) at (-\deltax,+\deltay) {};
\node (mult_label_outr) at (+\deltax,+\deltay) {};
\node [dot, fill=\internalclassicalStructColour] (mult) at (0,0) {};
\node (mult_label_in) at (0,-\deltay) {};
\draw[-] [in=270,out=135] (mult) to (mult_label_outl);
\draw[-] [in=270,out=45] (mult) to (mult_label_outr);
\draw[-] (mult_label_in) to (mult);

\end{tikzpicture}}\!}

	\tikzset{->-/.style={decoration={markings,mark=at position #1 with {\arrow{>}}},postaction={decorate}}}
	\tikzset{-<-/.style={decoration={markings,mark=at position #1 with {\arrow{<}}},postaction={decorate}}}
	\tikzset{dotted->-/.style={dotted,decoration={markings,mark=at position #1 with {\arrow{>}}},postaction={decorate}}}
	\tikzset{dotted-<-/.style={dotted,decoration={markings,mark=at position #1 with {\arrow{<}}},postaction={decorate}}}

\tikzstyle{env}=[copoint,regular polygon rotate=0,minimum width=0.2cm, fill=black]

\tikzstyle{probs}=[shape=semicircle,fill=white,draw=black,shape border rotate=180,minimum width=1.2cm]

\tikzstyle{every picture}=[baseline=-0.25em,scale=0.5]
\tikzstyle{dotpic}=[] 
\tikzstyle{diredges}=[every to/.style={diredge}]
\tikzstyle{math matrix}=[matrix of math nodes,left delimiter=(,right delimiter=),inner sep=2pt,column sep=1em,row sep=0.5em,nodes={inner sep=0pt},text height=1.5ex, text depth=0.25ex]

\tikzstyle{inline text}=[text height=1.5ex, text depth=0.25ex,yshift=0.5mm]
\tikzstyle{label}=[font=\footnotesize,text height=1.5ex, text depth=0.25ex,yshift=0.5mm]
\tikzstyle{left label}=[label,anchor=east,xshift=1.5mm]
\tikzstyle{right label}=[label,anchor=west,xshift=-1.5mm]

\tikzstyle{braceedge}=[decorate,decoration={brace,amplitude=2mm,raise=-1mm}]
\tikzstyle{small braceedge}=[decorate,decoration={brace,amplitude=1mm,raise=-1mm}]

\tikzstyle{doubled}=[line width=1.6pt] 
\tikzstyle{boldedge}=[doubled,shorten <=-0.17mm,shorten >=-0.17mm]
\tikzstyle{boldedgegray}=[doubled,gray,shorten <=-0.17mm,shorten >=-0.17mm]

\tikzstyle{semidoubled}=[line width=1.4pt] 
\tikzstyle{semiboldedgegray}=[semidoubled,gray,shorten <=-0.17mm,shorten >=-0.17mm]

\tikzstyle{boldedgedashed}=[very thick,dashed,shorten <=-0.17mm,shorten >=-0.17mm]
\tikzstyle{vboldedgedashed}=[doubled,dashed,shorten <=-0.17mm,shorten >=-0.17mm]
\tikzstyle{left hook arrow}=[left hook-latex]
\tikzstyle{right hook arrow}=[right hook-latex]
\tikzstyle{sembracket}=[line width=0.5pt,shorten <=-0.07mm,shorten >=-0.07mm]

\tikzstyle{causal edge}=[->,thick,gray]
\tikzstyle{causal nondir}=[thick,gray]
\tikzstyle{timeline}=[thick,gray, dashed]

\tikzstyle{cedge}=[<->,thick,gray!70!white]

\tikzstyle{empty diagram}=[draw=gray!40!white,dashed,shape=rectangle,minimum width=1cm,minimum height=1cm]
\tikzstyle{empty diagram small}=[draw=gray!50!white,dashed,shape=rectangle,minimum width=0.6cm,minimum height=0.5cm]

\tikzstyle{dot}=[inner sep=0mm,minimum width=3mm,minimum height=3mm,draw,shape=circle,text depth=-0.1mm]
\tikzstyle{smalldot}=[inner sep=0mm,minimum width=2mm,minimum height=2mm,draw,shape=circle,text depth=-0.1mm]
\tikzstyle{ddot}=[inner sep=0mm, doubled, minimum width=3.5mm,minimum height=3.5mm,draw,shape=circle]

\tikzstyle{black dot}=[dot,fill=black]
\tikzstyle{white dot}=[dot,fill=white,,text depth=-0.2mm]
\tikzstyle{green dot}=[white dot] 
\tikzstyle{gray dot}=[dot,fill=gray!40!white,,text depth=-0.2mm]
\tikzstyle{red dot}=[gray dot] 

\tikzstyle{black ddot}=[ddot,fill=black]
\tikzstyle{white ddot}=[ddot,fill=white]
\tikzstyle{gray ddot}=[ddot,fill=gray!40!white]

\tikzstyle{gray edge}=[gray!40!white]

\tikzstyle{small dot}=[inner sep=0.5mm,minimum width=0pt,minimum height=0pt,draw,shape=circle]

\tikzstyle{small black dot}=[small dot,fill=black]
\tikzstyle{small white dot}=[small dot,fill=white]
\tikzstyle{small gray dot}=[small dot,fill=gray!40!white]

\tikzstyle{causal dot}=[inner sep=0.4mm,minimum width=0pt,minimum height=0pt,draw=white,shape=circle,fill=gray!40!white]

\tikzstyle{phase dimensions}=[minimum size=5mm,font=\footnotesize,rectangle,rounded corners=2.5mm,inner sep=0.2mm,outer sep=-2mm,text height=1ex, text depth=0.25ex, yshift=0.5mm]
\tikzstyle{dphase dimensions}=[phase dimensions]

\tikzstyle{phase dot}=[dot,phase dimensions]

\tikzstyle{white phase dot}=[dot,fill=white,phase dimensions]
\tikzstyle{white phase ddot}=[ddot,fill=white,dphase dimensions]

\tikzstyle{white rect ddot}=[draw=black,fill=white,doubled,minimum size=5mm,font=\footnotesize,rectangle,rounded corners=2.5mm,inner sep=0.2mm]
\tikzstyle{gray rect ddot}=[draw=black,fill=gray!40!white,doubled,minimum size=6mm,font=\footnotesize,rectangle,rounded corners=3mm]

\tikzstyle{gray phase dot}=[dot,fill=gray!40!white,phase dimensions]
\tikzstyle{gray phase ddot}=[ddot,fill=gray!40!white,dphase dimensions]
\tikzstyle{grey phase dot}=[gray phase dot]
\tikzstyle{grey phase ddot}=[gray phase ddot]

\tikzstyle{cnot}=[fill=white,shape=circle,inner sep=-1.4pt]
\tikzstyle{hadamard}=[square box,inner sep=0 pt,font=\footnotesize,minimum height=4mm,minimum width=4mm]
\tikzstyle{dhadamard}=[hadamard,doubled]
\tikzstyle{antipode}=[white dot,inner sep=0.3mm,font=\footnotesize]

\tikzstyle{scalar}=[diamond,draw,inner sep=0.5pt,font=\small]
\tikzstyle{dscalar}=[diamond,doubled, draw,inner sep=0.5pt,font=\small]

\tikzstyle{small box}=[rectangle,inline text,fill=white,draw,minimum height=5mm,yshift=-0.5mm,minimum width=5mm,font=\small]
\tikzstyle{small gray box}=[small box,fill=gray!30]
\tikzstyle{medium box}=[rectangle,inline text,fill=white,draw,minimum height=5mm,yshift=-0.5mm,minimum width=10mm,font=\small]
\tikzstyle{square box}=[small box] 
\tikzstyle{medium gray box}=[small box,fill=gray!30]
\tikzstyle{semilarge box}=[rectangle,inline text,fill=white,draw,minimum height=5mm,yshift=-0.5mm,minimum width=12.5mm,font=\small]
\tikzstyle{large box}=[rectangle,inline text,fill=white,draw,minimum height=5mm,yshift=-0.5mm,minimum width=15mm,font=\small]
\tikzstyle{large gray box}=[small box,fill=gray!30]

\tikzstyle{gray square point}=[small box,fill=gray!50]

\tikzstyle{dphase box white}=[dbox]
\tikzstyle{dphase box gray}=[dbox,fill=gray!50!white]

\tikzstyle{point}=[regular polygon,regular polygon sides=3,draw,scale=0.75,inner sep=-0.5pt,minimum width=9mm,fill=white,regular polygon rotate=180]
\tikzstyle{copoint}=[regular polygon,regular polygon sides=3,draw,scale=0.75,inner sep=-0.5pt,minimum width=9mm,fill=white]
\tikzstyle{dpoint}=[point,doubled]
\tikzstyle{dcopoint}=[copoint,doubled]

\tikzstyle{wide copoint}=[fill=white,draw,shape=isosceles triangle,shape border rotate=90,isosceles triangle stretches=true,inner sep=0pt,minimum width=1.5cm,minimum height=6.12mm]
\tikzstyle{wide point}=[fill=white,draw,shape=isosceles triangle,shape border rotate=-90,isosceles triangle stretches=true,inner sep=0pt,minimum width=1.5cm,minimum height=6.12mm,yshift=-0.0mm]
\tikzstyle{wide point plus}=[fill=white,draw,shape=isosceles triangle,shape border rotate=-90,isosceles triangle stretches=true,inner sep=0pt,minimum width=1.74cm,minimum height=7mm,yshift=-0.0mm]

\tikzstyle{wide dpoint}=[fill=white,doubled,draw,shape=isosceles triangle,shape border rotate=-90,isosceles triangle stretches=true,inner sep=0pt,minimum width=1.5cm,minimum height=6.12mm,yshift=-0.0mm]

\tikzstyle{tinypoint}=[regular polygon,regular polygon sides=3,draw,scale=0.55,inner sep=-0.15pt,minimum width=6mm,fill=white,regular polygon rotate=180] 

\tikzstyle{white point}=[point]
\tikzstyle{white dpoint}=[dpoint]
\tikzstyle{green point}=[white point] 
\tikzstyle{white copoint}=[copoint]
\tikzstyle{gray point}=[point,fill=gray!40!white]
\tikzstyle{gray dpoint}=[gray point,doubled]
\tikzstyle{red point}=[gray point] 
\tikzstyle{gray copoint}=[copoint,fill=gray!40!white]
\tikzstyle{gray dcopoint}=[gray copoint,doubled]

\tikzstyle{black point}=[point,fill=black]
\tikzstyle{black copoint}=[copoint,fill=black]

\tikzstyle{tiny gray point}=[tinypoint,fill=gray!40!white]

\tikzstyle{diredge}=[->]
\tikzstyle{rdiredge}=[<-]
\tikzstyle{thickdiredge}=[->, very thick]
\tikzstyle{pointer edge}=[->,very thick,gray]
\tikzstyle{pointer edge part}=[very thick,gray]
\tikzstyle{dashed edge}=[dashed]
\tikzstyle{thick dashed edge}=[very thick,dashed]
\tikzstyle{thick gray dashed edge}=[thick dashed edge,gray!40]
\tikzstyle{thick map edge}=[very thick,|->]

\makeatletter
\newcommand{\boxshape}[3]{%
\pgfdeclareshape{#1}{
\inheritsavedanchors[from=rectangle] 
\inheritanchorborder[from=rectangle]
\inheritanchor[from=rectangle]{center}
\inheritanchor[from=rectangle]{north}
\inheritanchor[from=rectangle]{south}
\inheritanchor[from=rectangle]{west}
\inheritanchor[from=rectangle]{east}
\backgroundpath{
\southwest \pgf@xa=\pgf@x \pgf@ya=\pgf@y
\northeast \pgf@xb=\pgf@x \pgf@yb=\pgf@y

\@tempdima=#2
\@tempdimb=#3

\pgfpathmoveto{\pgfpoint{\pgf@xa - 5pt + \@tempdima}{\pgf@ya}}
\pgfpathlineto{\pgfpoint{\pgf@xa - 5pt - \@tempdima}{\pgf@yb}}
\pgfpathlineto{\pgfpoint{\pgf@xb + 5pt + \@tempdimb}{\pgf@yb}}
\pgfpathlineto{\pgfpoint{\pgf@xb + 5pt - \@tempdimb}{\pgf@ya}}
\pgfpathlineto{\pgfpoint{\pgf@xa - 5pt + \@tempdima}{\pgf@ya}}
\pgfpathclose
}
}}

\boxshape{NEbox}{0pt}{5pt}
\boxshape{SEbox}{0pt}{-5pt}
\boxshape{NWbox}{5pt}{0pt}
\boxshape{SWbox}{-5pt}{0pt}
\boxshape{EBox}{-3pt}{3pt}
\boxshape{WBox}{3pt}{-3pt}
\makeatother

\tikzstyle{cloud}=[shape=cloud,draw,minimum width=1.5cm,minimum height=1.5cm]

\tikzstyle{map}=[draw,shape=NEbox,inner sep=2pt,minimum height=6mm,fill=white]
\tikzstyle{dashedmap}=[draw,dashed,shape=NEbox,inner sep=2pt,minimum height=6mm,fill=white]
\tikzstyle{mapdag}=[draw,shape=SEbox,inner sep=2pt,minimum height=6mm,fill=white]
\tikzstyle{mapadj}=[draw,shape=SEbox,inner sep=2pt,minimum height=6mm,fill=white]
\tikzstyle{maptrans}=[draw,shape=SWbox,inner sep=2pt,minimum height=6mm,fill=white]
\tikzstyle{mapconj}=[draw,shape=NWbox,inner sep=2pt,minimum height=6mm,fill=white]

\tikzstyle{medium map}=[draw,shape=NEbox,inner sep=2pt,minimum height=6mm,fill=white,minimum width=7mm]
\tikzstyle{medium map dag}=[draw,shape=SEbox,inner sep=2pt,minimum height=6mm,fill=white,minimum width=7mm]
\tikzstyle{medium map adj}=[draw,shape=SEbox,inner sep=2pt,minimum height=6mm,fill=white,minimum width=7mm]
\tikzstyle{medium map trans}=[draw,shape=SWbox,inner sep=2pt,minimum height=6mm,fill=white,minimum width=7mm]
\tikzstyle{medium map conj}=[draw,shape=NWbox,inner sep=2pt,minimum height=6mm,fill=white,minimum width=7mm]
\tikzstyle{semilarge map}=[draw,shape=NEbox,inner sep=2pt,minimum height=6mm,fill=white,minimum width=9.5mm]
\tikzstyle{semilarge map trans}=[draw,shape=SWbox,inner sep=2pt,minimum height=6mm,fill=white,minimum width=9.5mm]
\tikzstyle{semilarge map adj}=[draw,shape=SEbox,inner sep=2pt,minimum height=6mm,fill=white,minimum width=9.5mm]
\tikzstyle{semilarge map dag}=[draw,shape=SEbox,inner sep=2pt,minimum height=6mm,fill=white,minimum width=9.5mm]
\tikzstyle{semilarge map conj}=[draw,shape=NWbox,inner sep=2pt,minimum height=6mm,fill=white,minimum width=9.5mm]
\tikzstyle{large map}=[draw,shape=NEbox,inner sep=2pt,minimum height=6mm,fill=white,minimum width=12mm]
\tikzstyle{very large map}=[draw,shape=NEbox,inner sep=2pt,minimum height=6mm,fill=white,minimum width=17mm]

\tikzstyle{medium dmap}=[draw,doubled,shape=NEbox,inner sep=2pt,minimum height=6mm,fill=white,minimum width=7mm]
\tikzstyle{medium dmap dag}=[draw,doubled,shape=SEbox,inner sep=2pt,minimum height=6mm,fill=white,minimum width=7mm]
\tikzstyle{medium dmap adj}=[draw,doubled,shape=SEbox,inner sep=2pt,minimum height=6mm,fill=white,minimum width=7mm]
\tikzstyle{medium dmap trans}=[draw,doubled,shape=SWbox,inner sep=2pt,minimum height=6mm,fill=white,minimum width=7mm]
\tikzstyle{medium dmap conj}=[draw,doubled,shape=NWbox,inner sep=2pt,minimum height=6mm,fill=white,minimum width=7mm]
\tikzstyle{semilarge dmap}=[draw,doubled,shape=NEbox,inner sep=2pt,minimum height=6mm,fill=white,minimum width=9.5mm]
\tikzstyle{semilarge dmap trans}=[draw,doubled,shape=SWbox,inner sep=2pt,minimum height=6mm,fill=white,minimum width=9.5mm]
\tikzstyle{semilarge dmap adj}=[draw,doubled,shape=SEbox,inner sep=2pt,minimum height=6mm,fill=white,minimum width=9.5mm]
\tikzstyle{semilarge dmap dag}=[draw,doubled,shape=SEbox,inner sep=2pt,minimum height=6mm,fill=white,minimum width=9.5mm]
\tikzstyle{semilarge dmap conj}=[draw,doubled,shape=NWbox,inner sep=2pt,minimum height=6mm,fill=white,minimum width=9.5mm]
\tikzstyle{large dmap}=[draw,doubled,shape=NEbox,inner sep=2pt,minimum height=6mm,fill=white,minimum width=12mm]
\tikzstyle{large dmap conj}=[draw,doubled,shape=NWbox,inner sep=2pt,minimum height=6mm,fill=white,minimum width=12mm]
\tikzstyle{large dmap trans}=[draw,doubled,shape=SWbox,inner sep=2pt,minimum height=6mm,fill=white,minimum width=12mm]
\tikzstyle{very large dmap}=[draw,doubled,shape=NEbox,inner sep=2pt,minimum height=6mm,fill=white,minimum width=19.5mm]

\tikzstyle{muxbox}=[draw,shape=rectangle,minimum height=3mm,minimum width=3mm,fill=white]
\tikzstyle{dmuxbox}=[muxbox,doubled]

\tikzstyle{box}=[draw,shape=rectangle,inner sep=2pt,minimum height=6mm,minimum width=6mm,fill=white]
\tikzstyle{dbox}=[draw,doubled,shape=rectangle,inner sep=2pt,minimum height=6mm,minimum width=6mm,fill=white]
\tikzstyle{dmap}=[draw,doubled,shape=NEbox,inner sep=2pt,minimum height=6mm,fill=white]
\tikzstyle{dmapdag}=[draw,doubled,shape=SEbox,inner sep=2pt,minimum height=6mm,fill=white]
\tikzstyle{dmapadj}=[draw,doubled,shape=SEbox,inner sep=2pt,minimum height=6mm,fill=white]
\tikzstyle{dmaptrans}=[draw,doubled,shape=SWbox,inner sep=2pt,minimum height=6mm,fill=white]
\tikzstyle{dmapconj}=[draw,doubled,shape=NWbox,inner sep=2pt,minimum height=6mm,fill=white]

\tikzstyle{ddmap}=[draw,doubled,dashed,shape=NEbox,inner sep=2pt,minimum height=6mm,fill=white]
\tikzstyle{ddmapdag}=[draw,doubled,dashed,shape=SEbox,inner sep=2pt,minimum height=6mm,fill=white]
\tikzstyle{ddmapadj}=[draw,doubled,dashed,shape=SEbox,inner sep=2pt,minimum height=6mm,fill=white]
\tikzstyle{ddmaptrans}=[draw,doubled,dashed,shape=SWbox,inner sep=2pt,minimum height=6mm,fill=white]
\tikzstyle{ddmapconj}=[draw,doubled,dashed,shape=NWbox,inner sep=2pt,minimum height=6mm,fill=white]

\boxshape{sNEbox}{0pt}{3pt}
\boxshape{sSEbox}{0pt}{-3pt}
\boxshape{sNWbox}{3pt}{0pt}
\boxshape{sSWbox}{-3pt}{0pt}
\tikzstyle{smap}=[draw,shape=sNEbox,fill=white]
\tikzstyle{smapdag}=[draw,shape=sSEbox,fill=white]
\tikzstyle{smapadj}=[draw,shape=sSEbox,fill=white]
\tikzstyle{smaptrans}=[draw,shape=sSWbox,fill=white]
\tikzstyle{smapconj}=[draw,shape=sNWbox,fill=white]

\tikzstyle{dsmap}=[draw,dashed,shape=sNEbox,fill=white]
\tikzstyle{dsmapdag}=[draw,dashed,shape=sSEbox,fill=white]
\tikzstyle{dsmaptrans}=[draw,dashed,shape=sSWbox,fill=white]
\tikzstyle{dsmapconj}=[draw,dashed,shape=sNWbox,fill=white]

\boxshape{mNEbox}{0pt}{10pt}
\boxshape{mSEbox}{0pt}{-10pt}
\boxshape{mNWbox}{10pt}{0pt}
\boxshape{mSWbox}{-10pt}{0pt}
\tikzstyle{mmap}=[draw,shape=mNEbox]
\tikzstyle{mmapdag}=[draw,shape=mSEbox]
\tikzstyle{mmaptrans}=[draw,shape=mSWbox]
\tikzstyle{mmapconj}=[draw,shape=mNWbox]

\tikzstyle{mmapgray}=[draw,fill=gray!40!white,shape=mNEbox]
\tikzstyle{smapgray}=[draw,fill=gray!40!white,shape=sNEbox]

\makeatletter
\pgfdeclareshape{cornerpoint}{
\inheritsavedanchors[from=rectangle] 
\inheritanchorborder[from=rectangle]
\inheritanchor[from=rectangle]{center}
\inheritanchor[from=rectangle]{north}
\inheritanchor[from=rectangle]{south}
\inheritanchor[from=rectangle]{west}
\inheritanchor[from=rectangle]{east}
\backgroundpath{
\southwest \pgf@xa=\pgf@x \pgf@ya=\pgf@y
\northeast \pgf@xb=\pgf@x \pgf@yb=\pgf@y

\pgfmathsetmacro{\pgf@shorten@left}{\pgfkeysvalueof{/tikz/shorten left}}
\pgfmathsetmacro{\pgf@shorten@right}{\pgfkeysvalueof{/tikz/shorten right}}

\pgfpathmoveto{\pgfpoint{0.5 * (\pgf@xa + \pgf@xb)}{\pgf@ya - 5pt}}
\pgfpathlineto{\pgfpoint{\pgf@xa - 8pt + \pgf@shorten@left}{\pgf@yb - 1.5 * \pgf@shorten@left}}
\pgfpathlineto{\pgfpoint{\pgf@xa - 8pt + \pgf@shorten@left}{\pgf@yb}}
\pgfpathlineto{\pgfpoint{\pgf@xb + 8pt - \pgf@shorten@right}{\pgf@yb}}
\pgfpathlineto{\pgfpoint{\pgf@xb + 8pt - \pgf@shorten@right}{\pgf@yb - 1.5 * \pgf@shorten@right}}
\pgfpathclose
}
}

\pgfdeclareshape{cornercopoint}{
\inheritsavedanchors[from=rectangle] 
\inheritanchorborder[from=rectangle]
\inheritanchor[from=rectangle]{center}
\inheritanchor[from=rectangle]{north}
\inheritanchor[from=rectangle]{south}
\inheritanchor[from=rectangle]{west}
\inheritanchor[from=rectangle]{east}
\backgroundpath{
\southwest \pgf@xa=\pgf@x \pgf@ya=\pgf@y
\northeast \pgf@xb=\pgf@x \pgf@yb=\pgf@y

\pgfmathsetmacro{\pgf@shorten@left}{\pgfkeysvalueof{/tikz/shorten left}}
\pgfmathsetmacro{\pgf@shorten@right}{\pgfkeysvalueof{/tikz/shorten right}}

\pgfpathmoveto{\pgfpoint{0.5 * (\pgf@xa + \pgf@xb)}{\pgf@yb + 5pt}}
\pgfpathlineto{\pgfpoint{\pgf@xa - 8pt + \pgf@shorten@left}{\pgf@ya + 1.5 * \pgf@shorten@left}}
\pgfpathlineto{\pgfpoint{\pgf@xa - 8pt + \pgf@shorten@left}{\pgf@ya}}
\pgfpathlineto{\pgfpoint{\pgf@xb + 8pt - \pgf@shorten@right}{\pgf@ya}}
\pgfpathlineto{\pgfpoint{\pgf@xb + 8pt - \pgf@shorten@right}{\pgf@ya + 1.5 * \pgf@shorten@right}}
\pgfpathclose
}
}

\makeatother

\pgfkeyssetvalue{/tikz/shorten left}{0pt}
\pgfkeyssetvalue{/tikz/shorten right}{0pt}

\tikzstyle{kpoint common}=[draw,fill=white,inner sep=1pt,minimum height=3mm]
\tikzstyle{kpoint}=[shape=cornerpoint,shorten left=5pt,kpoint common]
\tikzstyle{kpoint adjoint}=[shape=cornercopoint,shorten left=5pt,kpoint common]
\tikzstyle{kpoint conjugate}=[shape=cornerpoint,shorten right=5pt,kpoint common]
\tikzstyle{kpoint transpose}=[shape=cornercopoint,shorten right=5pt,kpoint common]
\tikzstyle{kpoint symm}=[shape=cornerpoint,shorten left=5pt,shorten right=5pt,kpoint common]

\tikzstyle{black kpoint}=[shape=cornerpoint,shorten left=5pt,kpoint common,fill=black]
\tikzstyle{black kpoint adjoint}=[shape=cornercopoint,shorten left=5pt,kpoint common,fill=black]

\tikzstyle{kpointdag}=[kpoint adjoint]
\tikzstyle{kpointadj}=[kpoint adjoint]
\tikzstyle{kpointconj}=[kpoint conjugate]
\tikzstyle{kpointtrans}=[kpoint transpose]

\tikzstyle{big kpoint}=[kpoint, minimum width=1.2 cm, minimum height=8mm, inner sep=4pt, text depth=3mm]

\tikzstyle{wide kpoint}=[kpoint, minimum width=1 cm, inner sep=2pt, text depth=-0.7 mm]
\tikzstyle{wide kpointdag}=[kpointdag, minimum width=1 cm, inner sep=2pt, text depth=0.7 mm]
\tikzstyle{wide kpointconj}=[kpointconj, minimum width=1 cm, inner sep=2pt, text depth=-0.7 mm]
\tikzstyle{wide kpointtrans}=[kpointtrans, minimum width=1 cm, inner sep=2pt, text depth=0.7 mm]

\tikzstyle{gray kpoint}=[kpoint,fill=gray!50!white]
\tikzstyle{gray kpointdag}=[kpointdag,fill=gray!50!white]
\tikzstyle{gray kpointadj}=[kpointadj,fill=gray!50!white]
\tikzstyle{gray kpointconj}=[kpointconj,fill=gray!50!white]
\tikzstyle{gray kpointtrans}=[kpointtrans,fill=gray!50!white]

\tikzstyle{gray dkpoint}=[kpoint,fill=gray!50!white,doubled]
\tikzstyle{gray dkpointdag}=[kpointdag,fill=gray!50!white,doubled]
\tikzstyle{gray dkpointadj}=[kpointadj,fill=gray!50!white,doubled]
\tikzstyle{gray dkpointconj}=[kpointconj,fill=gray!50!white,doubled]
\tikzstyle{gray dkpointtrans}=[kpointtrans,fill=gray!50!white,doubled]

\tikzstyle{white label}=[draw,fill=white,rectangle,inner sep=0.7 mm]
\tikzstyle{gray label}=[draw,fill=gray!50!white,rectangle,inner sep=0.7 mm]
\tikzstyle{black label}=[draw,fill=black,rectangle,inner sep=0.7 mm]

\tikzstyle{dkpoint}=[kpoint,doubled]
\tikzstyle{wide dkpoint}=[wide kpoint,doubled]
\tikzstyle{dkpointdag}=[kpoint adjoint,doubled]
\tikzstyle{dkcopoint}=[kpoint adjoint,doubled]
\tikzstyle{dkpointadj}=[kpoint adjoint,doubled]
\tikzstyle{dkpointconj}=[kpoint conjugate,doubled]
\tikzstyle{dkpointtrans}=[kpoint transpose,doubled]

\tikzstyle{kscalar}=[kpoint common, shape=EBox, inner xsep=-1pt, inner ysep=3pt,font=\small]
\tikzstyle{kscalarconj}=[kpoint common, shape=WBox, inner xsep=-1pt, inner ysep=3pt,font=\small]

 \tikzstyle{upground}=[circuit ee IEC,thick,ground,rotate=90,scale=2.5]
 \tikzstyle{downground}=[circuit ee IEC,thick,ground,rotate=-90,scale=2.5]
 \tikzstyle{bigground}=[regular polygon,regular polygon sides=3,draw=gray,scale=0.50,inner sep=-0.5pt,minimum width=10mm,fill=gray]

\tikzstyle{arrs}=[-latex,font=\small,auto]
\tikzstyle{arrow plain}=[arrs]
\tikzstyle{arrow dashed}=[dashed,arrs]
\tikzstyle{arrow bold}=[very thick,arrs]
\tikzstyle{arrow hide}=[draw=white!0,-]
\tikzstyle{arrow reverse}=[latex-]
\tikzstyle{cdnode}=[]

\title{A Bestiary of Sets and Relations}
\author{
	Stefano Gogioso
	\institute{Quantum Group \\ University of Oxford}
	\email{stefano.gogioso@cs.ox.ac.uk}
}
\def\titlerunning{A Bestiary of Sets and Relations}
\def\authorrunning{S. Gogioso}

\begin{document}
	\maketitle

	\begin{abstract}
		Building on established literature and recent developments in the graph-theoretic characterisation of its CPM category, we provide a treatment of pure state and mixed state quantum mechanics in the category $\fRelCategory$ of finite sets and relations. On the way, we highlight the wealth of exotic beasts that hide amongst the extensive operational and structural similarities that the theory shares with more traditional arenas of categorical quantum mechanics, such as the category $\fdHilbCategory$. We conclude our journey by proving that $\fRelCategory$ is local, but not without some unexpected twists.
	\end{abstract}

\section{Introduction}
	\label{section_Background}

	The Categorical Quantum Mechanics programme \cite{CQM-seminal}\cite{CQM-QCSnotes}\cite{CQM-CQMnotes} is concerned with the understanding, through the language of dagger symmetric monoidal categories, of the structural and operational features of quantum theory. The investigation of classical-quantum duality goes through the definition of classical structures, \cite{CQM-QuantumClassicalStructuralism}\cite{CQM-OrthogonalBases} a.k.a. special commutative $\dagger$-Frobenius algebras, which play a central role as the abstract incarnation of non-degenerate observables, and lie at the very heart of the paradigm. The dagger allows for an abstract definition of unitarity, while the CPM construction from \cite{CQM-SelingerCPM} can be leveraged to rigorously define \cite{CQM-EvironmentClassicalChannels}\cite{CQM-QCSnotes}\cite{CQM-QuantumMeasuNoSums} discarding maps, mixed states, decoherence and measurements. 

	The extreme versatility of the approach has given birth, in the years, to many a toy model of quantum theory, a particularly interesting one being the category $\fRelCategory$ of finite sets and relations   The apparent classicality-by-construction of $\fRelCategory$ contrasts starkly with the presence of many trademarks of quantum theory: superposition, entanglement, plenty of classical structures and phases. The full characterisation of classical structures in terms of abelian groupoids is known \cite{BestRel-ClassicalStructuresRel}\cite{BestRel-RelMutuallyUnbiased} (and generalised in \cite{BestRel-RelativeFrob} to arbitrary groupoids and special $\dagger$-Frobenius algebras), and provides a stimulating playground in which to stress-test several operational features \cite{BestRel-ZengAlgoRel}\cite{StefanoGogioso-RepTheoryCQM}\cite{BestRel-ToyQuantumCategories}\cite{CQM-ZXCalculusSeminal}\cite{CQM-QuantumClassicalStructuralism} taken for granted in more the traditional arenas of quantum mechanics. $\fRelCategory$ has recently found application in Natural Language Processing \cite{BestRel-CPMNLP}, where it relates to Montague Semantics; also in the context of NLP, its CPM category has been shown \cite{BestRel-graphCPMrel} to have a particularly handy graph-theoretic characterisation. 

	In this work, we give an overview of $\fRelCategory$ as a model of categorical quantum mechanics: comparing and contrasting with the category $\fdHilbCategory$ of Hilbert spaces and linear maps, we highlight the many similarities with the traditional framework, and a number of fairly traits unique to $\fRelCategory$. Using the newly developed \cite{BestRel-graphCPMrel} graph-theoretic characterisation of $\CPMCategory{\fRelCategory}$, we explore the exotic landscape of mixed state quantum mechanics in $\fRelCategory$: we characterise decoherence maps and demolition measurements, and we manage to show that, despite the significant differences from traditional frameworks, the theory is local. This proof of locality, done with respect to operationally defined demolition measurements, extends and completes the one presented in \cite{BestRel-RelationalNonlocality}, which only applies to measurements valued in the discrete classical structure (the only one with enough classical points). In order to improve the flow, we have omitted the proofs of results well established in the literature provided and/or straightforward to check if necessary.

\section{Pure state quantum mechanics in $\fRelCategory$}
	\label{section_QMRel}

	The category $\fRelCategory$ is defined as having finite sets as its objects, relations $R \subseteq X \times Y$ as morphisms $X \rightarrow Y$ and relational composition. As $\fdHilbCategory$, it is a $\dagger$ symmetric monoidal category (henceforth $\dagger$-SMC), with the cartesian product $\times$ of sets/relations as tensor, the singleton set $1 := \{\star\}$ as tensor unit and a dagger defined by:
	\begin{equation}
		R^\dagger := \suchthat{(y,x)}{(x,y)\in R}
	\end{equation}

	\noindent The scalars of $\fRelCategory$ are $\{\bot,\top\}$, with the tensor and unitors inducing multiplication $\times$ (or equivalently $\inf$) on them. As $\fdHilbCategory$, the category $\fRelCategory$ is enriched over finite monoids, with the tensor distributing over a superposition operation $\vee$, the union of relations.\footnote{We will denote union and intersections of subets/relations by $\vee$ and $\wedge$, to avoid confusion with the cup $\cup$ and cap $\cap$ of the compact structure defined later. It also bodes well with potential generalisations from booleans to semirings or locales.} The scalars form a semiring $(\{\bot,\top\},\vee,\times)$, which opens up the doors for the application of methods from sheaf-theoretic non-locality.\cite{NLC-SheafSeminal}

	Morphisms $R: X \rightarrow Y$, seen as subsets $R \subseteq X \times Y$, form a sup-semilattice under $\vee$.\footnote{In fact, they form a complete distributive boolean lattice, with intersection of relations $\wedge$ and complement of relations $\neg$.} This applies in particular to states, which can be seen as subsets $\psi \subseteq X$: their hierarchical superposition structure, with $\ket{X}$ as maximum and the elements $\ket{x}$ for $x\in X$ as atoms, is the first significant difference from $\fdHilbCategory$. From now on, we will denote the elements of $X$ by $x\in X$, and the states/subsets of $X$ equivalently by $\psi \subseteq X$ or $\ket{\psi} : 1 \rightarrow X$.
	
	Just as in $\fdHilbCategory$, the tensor of $\fRelCategory$ is not a cartesian tensor in the categorical sense:\footnote{Despite being called the \textit{cartesian product} of relations.} it is sufficient to observe that $\fRelCategory$ has lots of entangled states (letting $X$ and $Y$ have $n$ and $m$ elements respectively, there are $2^{n+m-2}+1$ separable states of $X \times Y$, out of $2^{nm}$ states). The central role in CQM, however, is played by classical structures, rather than states, and we now move to their characterisation in $\fRelCategory$ as abelian groupoids.

	\begin{definition}
		An \textbf{abelian groupoid} on a finite set $X$ is a disjoint family $(G_\lambda,+_\lambda,0_\lambda)_{\lambda\in \Lambda}$ of abelian groups such that $\vee_{\lambda \in  \Lambda} G_\lambda = X$. 
	\end{definition}

	\noindent We will use notation $\oplus_{\lambda\in \Lambda} G_\lambda$ to denote one such abelian groupoid. Notice that any groupoid induces a partition of $X$ into states $G_\lambda \subseteq X$: when one such groupoid (and hence partition) is understood, we will label the elements of $X$ by $g_\lambda$, with $\lambda$ ranging over $\Lambda$ and $g_\lambda$ ranging over $G_\lambda$. Now we can turn our attention to classical structures.

	\begin{theorem}\label{thm_ClassicalStructsFRel}
		Classical structures in $\fRelCategory$ coincide with abelian groupoids, i.e.\ if $(\Zstructure)$ is a classical structure on $X$, then there exists a unique abelian groupoid $\oplus_{\lambda\in \Lambda} G_\lambda$ such that:
		\begin{align}
			\ZmultSym &= \suchthat{((g_\lambda,g'_{\lambda}),g_\lambda+ g'_\lambda)}{g,g'\in G_\lambda, \lambda \in  \Lambda} \\
			\ZunitSym &= \suchthat{(\star,0_\lambda)}{\lambda \in  \Lambda}\\
			\ZcomultSym &= \suchthat{(g_\lambda,(g'_\lambda,g''_\lambda))}{g,g',g''\in G_\lambda, \lambda \in  \Lambda \text{ and }g'_\lambda + g''_\lambda = g_\lambda} \\
			\ZcounitSym &= \suchthat{(0_\lambda,\star)}{\lambda \in  \Lambda}
		\end{align} 
		Furthermore, the classical points of $(\Zstructure)$ are exactly the states $\ket{G_\lambda} : 1 \rightarrow X$. The phases of $(\Zstructure)$ are exactly the states in the form $\suchthat{(\star,g_\lambda)}{\lambda\in \lambda}$ with $g_\lambda \in  G_\lambda$ for all $\lambda$ (so there are a lot more phases than classical points).
	\end{theorem} 
	\begin{proof}
		Proof of the classical structure / abelian groupoid connection is originally presented in \cite{BestRel-ClassicalStructuresRel}. The classical points are classified in \cite{BestRel-RelMutuallyUnbiased}. To the best of our knowledge, the phase structure in $\fRelCategory$ is first discussed in \cite{StefanoGogioso-MerminNonLocality}.
	\end{proof}

	\noindent We see that $\ZmultSym$ is a partial function acting as the group operation of $G_\lambda$ whenever its two arguments both belong to the same group $G_\lambda$, and being undefined otherwise. In the case of strongly complementary structures things get particularly interesting: we present the result here (from \cite{BestRel-RelMutuallyUnbiased}) for sake of completeness, but we will not concern ourselves with strong complementarity any further in this work, even though it is possible to use strong complementarity to implement Fourier transforms in $\fRelCategory$ (as shown in \cite{StefanoGogioso-RepTheoryCQM}). 
	
	\begin{theorem}\label{thm_StrongComplementarityFRel}
		Let $(\hbox{\begin{tikzpicture} [scale=1,transform shape] 

\def\deltax{0.3} 
\def\deltay{0.5} 


\node [dot, fill=\classicalStructColour] (mult) at (0,0) {};

\end{tikzpicture}}\!,\hbox{\begin{tikzpicture} [scale=1,transform shape] 

\def\deltax{0.3} 
\def\deltay{0.5} 


\node [dot, fill=\groupStructColour] (mult) at (0,0) {};

\end{tikzpicture}}\!)$ be a pair of strongly complementary classical structures\footnote{From now on we will use a dot of the structure colour to denote a classical structure.} in $\fRelCategory$. Then there exist unique groups $G,H$ such that $\hbox{\begin{tikzpicture} [scale=1,transform shape] 

\def\deltax{0.3} 
\def\deltay{0.5} 


\node [dot, fill=\classicalStructColour] (mult) at (0,0) {};

\end{tikzpicture}}\!$ corresponds to the groupoid $\oplus_{h\in H}G_h$ with $G_h \isom G$ for all $h\in H$ and $\hbox{\begin{tikzpicture} [scale=1,transform shape] 

\def\deltax{0.3} 
\def\deltay{0.5} 


\node [dot, fill=\groupStructColour] (mult) at (0,0) {};

\end{tikzpicture}}\!$ corresponds to the groupoid $\oplus_{g\in G}H_g$ with $H_g \isom H$ for all $g\in G$.   
	\end{theorem}

	\noindent By \textbf{morphisms of classical structures} we will mean homomorphisms of the comultiplications and counits. In any $\dagger$-SMC, they act as functions on the sets of classical points, and just as in $\fdHilbCategory$ there is a canonical way of seeing functions of the sets of classical points as morphisms of classical structures in $\fRelCategory$. 

	Given a pair of classical structures on sets $X$ and $Y$ corresponding to groupoids $\oplus_{\lambda \in  \Lambda} G_\lambda$ and $\oplus_{\gamma \in  \Gamma} H_\gamma$, and a partial function $f: \Lambda \rightharpoonup \Gamma$ of sets, we can construct the following morphism $R_f: X \rightarrow Y$ in $\fRelCategory$ which is a morphism of the given classical structures and acts as the required partial function on the classical points:
	\begin{equation} \label{eqn_embedingClassicalFunctions}
		R_f := \bigvee_{\lambda \in  \domain{f}} \; G_\lambda \times H_{f(\lambda)} 
	\end{equation}

	\noindent This is the equivalent of the following morphism in $\fdHilbCategory$:
	\begin{equation}
		R_f := \sum_{\lambda \in  \domain{f}} \; \ket{f(\lambda)}\bra{\lambda} 
	\end{equation}

	\noindent So in $\fRelCategory$, exactly as in $\fdHilbCategory$, there is a natural way of doing classical computation by fixing classical structures and using the $R_f$ above to construct the required morphisms. But unlike $\fdHilbCategory$, $\fRelCategory$ has a lot more morphisms of classical structures than that. For example, if $f: \Lambda \stackrel{\SetCategory}{\longrightarrow} \Gamma$ is as before and $\Phi_{\lambda}: G_\lambda \stackrel{\AbCategory}{\longrightarrow} H_{f(\lambda)}$ is a family of group isomorphisms, then the following relation (which is a partial function $X \rightarrow Y$) acts exactly as the relation $R_f$ (which is, in general, not a function at all) on the desired classical points:
	\begin{equation}
		g_\lambda \mapsto \Phi_\lambda(g_\lambda)
	\end{equation}

	\noindent This is a consequence of another, more fundamental difference between $\fRelCategory$ and $\fdHilbCategory$: most classical structures do not have enough classical points. In fact there is a unique classical structure on each set $X$ that does: it is the the \textbf{discrete structure}, given by the discrete groupoid $\oplus_{x\in X} 0_x$ and having the singletons $\ket{x}:=\{x\}$ of $X$ as its classical points. Indeed, the classical points of the discrete structure yield a resolution of the identity:
	\begin{equation}
		\bigvee_{x\in X} \ket{x}\bra{x} = \suchthat{(x,x)}{x\in X} = \id{X}
	\end{equation}

	\noindent When the classical structures on $X$ and $Y$ are the discrete structures, Equation \ref{eqn_embedingClassicalFunctions} provides the usual embedding of the category of finite sets and partial functions in $\fRelCategory$:
	\begin{equation}
		R_f = \suchthat{(x,f(x))}{x \in \domain{f}}
	\end{equation}

	\noindent In any $\dagger$-SMC, any classical structure on some system $X$ can be used to induce a cup $\cup_X$ and cap $\cap_X$ on $X$: the discrete structures in $\fRelCategory$ induce a natural family of cups and caps, and hence a compact-closed structure for $\fRelCategory$. 

	The cup on $X$ is given by the relation $\cup_X := \suchthat{(\star,(x,x))}{x\in X}$, while the cap is the partial map $\cap_X := X \times X \rightharpoonup 1$ sending $(x,x)$ to $\star$ for all $x \in X$ and undefined everywhere else. The resulting conjugation is trivial, i.e. $R^\star = R$, and transposition coincides with the dagger, i.e. $R^T = R^\dagger$. This is somewhat different from $\fdHilbCategory$, where the traditional compact closed structure is not induced by any specific classical structure.

	Finally, a central role in pure state quantum mechanics is also played by isometries and unitaries. Recall that an isometry in any $\dagger$-SMC is a morphism $f: X \rightarrow Y$ such that $f^\dagger \circ f = \id{X}$, and a unitary is a morphism $U$ such that both $U$ and $U^\dagger$ are isometries. In $\fdHilbCategory$, unitaries coincide with orthonormal change-of-basis transformations, i.e. bijective classical maps. This results in the following (straightforward) lemma.
	
	\begin{lemma}\label{lemma_IsometriesUnitarisFHilb}
		If $f: X \rightarrow Y$ is a morphism in $\fdHilbCategory$, then $f$ is an isometry if and only if there are classical structures on $X$ and $Y$ making $f$ into an injective classical map. Furthermore, $f$ is a unitary if and only if it is a bijective classical map.\footnote{In fact, if $f$ is an isometry then fixing any classical structure on $X$ means there is a classical structure on $Y$ making $f$ into an injective classical map. If $f$ is furthermore a unitary, the structure on $Y$ is also unique.}
	\end{lemma}
	\begin{proof} 
		Classical structures (special commutative $\dagger$-Frobenius algebras) in $\fdHilbCategory$ correspond to orthonormal bases by \cite{CQM-OrthogonalBases}: let's fix the orthonormal basis corresponding to a classical structure on $X$ and consider the matrix of $f$ in that basis. 

		Then $f$ is an isometry if and only if all column vectors are orthonormal, and any orthonormal basis including the column vectors as a subset will give a classical structure on $Y$ making $f$ an injective classical map. 

		Furthermore, $f$ is unitary if and only if the column vectors form an orthonormal basis, corresponding to a unique classical structure on $Y$.
	\end{proof}

	\noindent In $\fRelCategory$, one could hope for unitaries that are isomorphisms between arbitrary classical structures (as it happens in $\fdHilbCategory$), giving rise to a non-trivial interplay. Unfortunately, the condition of isometry/unitarity turns out to be a lot more restrictive in $\fRelCategory$ than it it in $\fdHilbCategory$, as the following lemma summarises.
	
	\begin{lemma}\label{lemma_IsometriesUnitarisFRel}
		If $f: X \rightarrow Y$ is a morphism in $\fRelCategory$, then $f$ is an isometry if and only $f^\dagger$ is a surjective partial function. Equivalently, $f$ is an isometry if and only if it is an injective classical map with respect to the discrete structure on $X$ and some classical structure on $Y$. Furthermore, $f$ is a unitary if and only if it is a bijection.\footnote{Note that this is NOT the same as a bijective classical map between classical structures.}
	\end{lemma}
	\begin{proof} 
		Let $f: X \rightarrow Y$ be a morphism in $\fRelCategory$, i.e. a relation $f \subseteq X \times Y$. The condition $f^\dagger \circ f = \id{X}$ amounts to the following equation:
		\begin{equation}
			\bigvee_{x \in \domain{f}} \suchthat{(x,z)}{ (x,y) \in f \text{ and } (z,y) \in f} = \suchthat{(x,x)}{x \in X}
		\end{equation}
		This is true if and only if both (i) $\domain{f} = X$ (i.e. $f^\dagger$ is surjective) and (ii) $(x,y) \in f$ and $(z,y) \in f$ imply $x=z$ (i.e. $f^\dagger$ is a partial function). Furthermore $f$ is unitary if and only if both $f$ and $f^\dagger$ are surjective partial functions, which happens if and only if $f$ is a bijection. If $f^\dagger$ is surjective partial function, then $f$ can always be written in the following form, where $f^{-1}(x)$ are disjoint subsets of $Y$ for all $x \in X$:
		\begin{equation}
			\bigvee_{x \in  X} \;\ket{f^{-1}(x)}\bra{x}
		\end{equation}  
		Then any classical structure on $Y$ including all $\ket{f^{-1}(x)}$ amongst its classical points will make $f$ into a classical injection from the discrete structure on $X$.
 	\end{proof}

	\noindent Thus unitaries in $\fRelCategory$ are exactly the bijective classical maps between discrete structures: pure state quantum mechanics in $\fRelCategory$ suddenly becomes quite boring. Let's now proceed to mixed state quantum mechanics in $\CPMCategory{\fRelCategory}$, in the hope that the peculiar measurement structure, resulting from the general lack of enough classical points, will spice things up.

	\begin{remark}
		One last point on the relationship between classical structures in $\fRelCategory$, before moving on to mixed state quantum mechanics. Classical structures on some space $X$ in any $\dagger$-SMC with a distributive superposition operation can be given a preorder by defining $\hbox{\begin{tikzpicture} [scale=1,transform shape] 

\def\deltax{0.3} 
\def\deltay{0.5} 


\node [dot, fill=\classicalStructColour] (mult) at (0,0) {};

\end{tikzpicture}}\!\leq\hbox{\begin{tikzpicture} [scale=1,transform shape] 

\def\deltax{0.3} 
\def\deltay{0.5} 


\node [dot, fill=\groupStructColour] (mult) at (0,0) {};

\end{tikzpicture}}\!$ if and only if classical points of $\hbox{\begin{tikzpicture} [scale=1,transform shape] 

\def\deltax{0.3} 
\def\deltay{0.5} 


\node [dot, fill=\groupStructColour] (mult) at (0,0) {};

\end{tikzpicture}}\!$ can be expressed as superpositions of classical points of $\hbox{\begin{tikzpicture} [scale=1,transform shape] 

\def\deltax{0.3} 
\def\deltay{0.5} 


\node [dot, fill=\classicalStructColour] (mult) at (0,0) {};

\end{tikzpicture}}\!$ (possibly multiplied by scalars). 

		In $\fdHilbCategory$, the preorder is an equivalence relation with a single equivalence class, as all structures have enough classical points, but in $\fRelCategory$ this is not so. Recall that the classical structure induced by groupoid $\oplus_{\lambda\in \Lambda} G_\lambda$ on a set $X$ yields a partition $(G_\lambda)_{\lambda\in \Lambda}$ of $X$: thus there is a surjective map of classical structures onto partitions, where classical structures corresponding to the same partition are exactly those having the same classical points. 

		The surjection exactly quotients away the equivalence classes in the preorder, which therefore descends to the partial order (in fact a lattice) on partitions of $X$: this is given by the \textit{refinement} partial order, with the partition into singletons as a minimum and the partition with only $X$ as the maximum. The discrete structure is the only one mapping to the singleton partition, and is therefore the unique minimum for the preorder on classical structures.
	\end{remark}

\newcommand{\CPMrightarrow}{\stackrel{CPM}{\longrightarrow}}

\section{Mixed state quantum mechanics in $\fRelCategory$}
	\label{section_MixQMRel}

	The fundamental observation of categorical quantum mechanics is that there are only a few ingredients needed for an abstract, operational characterisation of pure state quantum theory: states, a dagger for inner products, a (symmetric) tensor for joint systems, classical structures for classical computation, unitaries for dynamics, an optional enrichment of morphisms (with appropriate distributivity law for the tensor) giving some notion of superposition. 

	But an equally fundamental aspect of quantum theory, not immediately captured by this framework, is given by measurements and mixed states: in order to introduce them in $\fRelCategory$, we turn our attention to the associated CPM category $\CPMCategory{\fRelCategory}$.\\

	\noindent The CPM category $\CPMCategory{\fRelCategory}$ has the same objects of $\fRelCategory$, and morphisms $R: X \CPMrightarrow Y$ in $\CPMCategory{\fRelCategory}$ are exactly the morphisms $R: X \times X \rightarrow Y \times Y$ in $\fRelCategory$ taking the following form:
	\begin{equation}\label{CPMmorphisms}
		\hbox{\begin{tikzpicture}[node distance = 10mm]

\node (center) {};

\node (map) [mapconj] [right of = center, xshift = -4mm] {$f$};
\node (out) [above of = map, yshift = +0mm] {};
\node (in) [below of = map, yshift = -0mm] {};

\node (mapconj) [map] [left of = center, xshift = +4mm] {$f^\star$};
\node (outconj) [above of = mapconj, yshift = +0mm] {};
\node (inconj) [below of = mapconj, yshift = -0mm] {};

\begin{pgfonlayer}{background}
\draw[->-=.5,out=90,in=270] (in) to (map);
\draw[->-=.5,out=90,in=270] (map) to (out);
\draw[-<-=.5,out=90,in=270] (inconj) to (mapconj);
\draw[-<-=.5,out=90,in=270] (mapconj) to (outconj);
\draw[->-=.53,out=90,in=90] (map.120) to (mapconj.60);
\end{pgfonlayer}

\end{tikzpicture}}
	\end{equation}
	
	\noindent If $f: X \rightarrow Z \times Y$, then we denote $f^\star = (f^{\dagger})^{T} = f$ in Diagram \ref{CPMmorphisms} as a morphism $f: X \rightarrow Y \times Z$ to keep the picture symmetric and avoid wire-crossing. We shall refer to morphisms in $\CPMCategory{\fRelCategory}$ as \textbf{CPM maps}, to CPM maps $1 \CPMrightarrow X$ as \textbf{mixed states} in $X$, and to CPM morphisms with no cap involved (i.e. with $f: X \rightarrow 1 \times Y$) as \textbf{pure maps} (or \textbf{pure states}, if $X=1$). 

	In particular, the caps are CPM maps $\cap_X : X \stackrel{CPM}{\longrightarrow} 1$ and the cups are mixed states $\cup_X : 1 \stackrel{CPM}{\longrightarrow} X$. Explicitly they are defined to be the following relations:
	\begin{align*}
		\cap_X &:= \suchthat{((x,x),\star)}{x \in X} \\
		\cup_X &:= \suchthat{(\star,(x,x))}{x \in X} \numberthis
	\end{align*}

	\noindent We will refer to $\cap_X$ as the \textbf{discarding map} on $X$, to post-composition with $\cap_X$ as \textbf{tracing out }$X$, and to $\cup_X$ as the \textbf{totally mixed state} in $X$. 

	Then general CPM maps are exactly obtained by tracing out some factor of the codomain of some pure map, or equivalently by applying some map to a totally mixed state in some factor of the domain: the conceptual importance of this observation comes from the following theorem (proven in \cite{CQM-SelingerCPM} and given operational interpretation in \cite{CQM-EvironmentClassicalChannels}).

	\begin{theorem}
		There is a faithful functor $I:\CategoryC \monom \CPMCategory{\CategoryC}$ of $\dagger$-SMCs from any compact-closed $\dagger$-SMC (with enough states) to the associated CPM Category $\CPMCategory{\CategoryC}$, itself a compact-closed $\dagger$-SMC, bijective on objects and mapping morphisms $f:X \rightarrow Y$ in $\CategoryC$ to the respective pure maps given by $f:X \rightarrow 1 \times Y$ in the notation of Diagram \ref{CPMmorphisms}. Thus the CPM construction can be seen, operationally, as the abstract process theory obtained from $\CategoryC$ by adding discarding maps (and/or totally mixed states).
	\end{theorem}

	\noindent In this setting, tracing out systems is interpreted as complete erasure of information about them, so it is no surprise that information-theoretic considerations come into play. Indeed, we will be interested in a particular class of CPM maps and mixed states: we say that a CPM map $R: X \CPMrightarrow Y$ (or mixed state, if $X=1$) is \textbf{causal} if $\cap_Y \cdot R = \cap_X$. 

	Causal CPM maps are one of the reason isometries (and unitaries in particular) in pure state quantum mechanics are so interesting: a CPM map in the form of Diagram \ref{CPMmorphisms} is causal if and only if $f: X \rightarrow Z \times Y$ is an isometry.\footnote{Observe that the cap on $Z \times Y$ is the product of the caps on $Z$ and on $Y$.} 

	In order to proceed with our investigation of $\CPMCategory{\fRelCategory}$, it is time to introduce the characterisation of mixed states and CPM maps in terms of graphs. Although the material here is fruit of the author's own work, full credit for the original characterisation goes to \cite{BestRel-graphCPMrel}, which contains all the details. 

	If $\rho : 1 \CPMrightarrow X$ is a mixed state, then it is straightforward to check that, seen as a relation $\rho \subseteq X \times X$, it has the following properties:
	\begin{enumerate}
		\item[(i)] $\rho$ is a symmetric relation, i.e. if $(x,y) \in  \rho$ then $(y,x) \in  \rho$
		\item[(ii)] $\rho$ is reflexive on all elements appearing in it, i.e. if $(x,y) \in  \rho$ then $(x,x)\in  \rho$ and $(y,y) \in  \rho$
	\end{enumerate}  

	\noindent Conversely, every relation with those properties is a mixed state. Furthermore, causal mixed states are exactly those with $\rho \neq \emptyset$. As a consequence of this characterisation, we can identify mixed states in $X$ with subgraphs of the complete graph $\completeGraph{X}$ with $X$ as set of nodes. The subgraph $\RelGraph{\rho} \leq \completeGraph{X}$ corresponding to a mixed state $\rho$ in $X$ is defined to have:
	\begin{enumerate}
		\item[(i)] nodes $\graphNode{x}$ specified by the pairs $(x,x) \in  \rho$ (corresponding to a subset of $X$)
		\item[(ii)] edges $\graphEdge{x}{y}$ specified by the pairs $(x,y) \in  \rho$ with $x \neq y$
	\end{enumerate}

	\noindent Furthermore, causal states correspond to the non-empty subgraphs. The graph characterisation of CPM maps can be then obtained using compact closure. The subgraph $\RelGraph{R} \leq \completeGraph{X \times Y}$ corresponding to a CPM map $R: X \CPMrightarrow Y$ is seen to have:
	\begin{enumerate}
		\item[(i)] nodes $\graphNode{[x,y]}$ specified by those pairs $x\in X$ and $y\in Y$ such that $\bra{y} \bra{y} R \ket{x} \ket{x} = 1$, where $R$ is seen as a $\fRelCategory$ morphism $R : X \times X \rightarrow Y \times Y$
		\item[(ii)] edges $\graphEdge{[x,y]}{[x',y']}$ specified by quadruplets $x,x'\in X$ and $y,y'\in Y$ with $\bra{y} \bra{y'} R \ket{x} \ket{x'} = 1$, where $R$ is seen as an $\fRelCategory$ morphism $R : X \times X \rightarrow Y \times Y$
	\end{enumerate}

	\noindent Given a CPM morphism $R: X \CPMrightarrow Y$ and a mixed state $\rho$ in $X$, it is interesting to characterise the subgraph $\RelGraph{R \cdot \rho} \leq \completeGraph{Y}$ of the mixed state $R \cdot \rho$ in terms of the subgraphs $\RelGraph{R} \leq \completeGraph{X \times Y}$ and $\RelGraph{\rho} \leq \completeGraph{X}$ of $R$ and $\rho$:
	\begin{enumerate}
		\item[(i)] if $\graphNode{x}$ is a node of $\RelGraph{\rho}$ and $\graphNode{[x,y]}$ is a node of $\RelGraph{R}$, then $\graphNode{y}$ is a node of $\RelGraph{R \cdot \rho}$
		\item[(ii)] if $\graphEdge{x}{x'}$ is an edge of $\RelGraph{\rho}$ and $\graphEdge{[x,y]}{[x',y']}$ is an edge of $\RelGraph{R}$, then $\graphEdge{y}{y'}$ is an edge of $\RelGraph{R \cdot \rho}$
	\end{enumerate}

	\noindent A generalisation of the argument above can be used to characterise the subgraph $\RelGraph{S \cdot R} \leq \completeGraph{X \times Z}$ of the composition of two CPM maps $R: X \CPMrightarrow Y$ and $S: Y \CPMrightarrow Z$:
	\begin{enumerate}
		\item[(i)] if $\graphNode{[x,y]}$ is a node of $\RelGraph{R}$ and $\graphNode{[y,z]}$ is a node of $\RelGraph{R}$, then $\graphNode{[x,z]}$ is a node of $\RelGraph{S \cdot R}$
		\item[(ii)] if $\graphEdge{[x,y]}{[x',y']}$ is an edge of $\RelGraph{R}$ and $\graphEdge{[y,z]}{[y',z']}$ is an edge of $\RelGraph{R}$, then $\graphEdge{[x,z]}{[x',z']}$ is an edge of $\RelGraph{S \cdot R}$
	\end{enumerate}

	\noindent A few examples of CPM maps and mixed states will give us a hands-on  understanding of this graph-theoretic characterisation. 
	\begin{enumerate}
		\item[1.] If $\rho$ is a pure state in $X$, corresponding to a $\fRelCategory$ state $S \subseteq X$, then $\RelGraph{\rho} \leq \completeGraph{X}$ is the clique on $S$, because $\rho = \suchthat{(s,s')}{s,s' \in  S}$; conversely, all cliques are pure states. 
		\item[2.] The pure map $\id{X} : X \CPMrightarrow X$, seen as the morphism $\id{X \times X} : X \times X \rightarrow X \times X$ in $\fRelCategory$, satisfies $\bra{y}\bra{y'} \id{X \times X} \ket{x}\ket{x'} = 1$ if and only if $x=y$ and $x'=y'$; as a consequence, the subgraph $\RelGraph{\id{X}} \leq \completeGraph{X \times X}$ is the clique on the diagonal $\Delta_X := \suchthat{(x,x)}{x \in  X} \subset X \times X$.
		\item[3.] The discarding map $\cap_X : X \CPMrightarrow 1$ corresponds to the discrete graph with node set $X \times 1 \isom X$ and no edges. Similarly, the totally mixed state in $X$ corresponds to the discrete graph with node set $X$. 
		\item[4.] If $G$ is a subgraph of $\completeGraph{X \times Y}$, define $\pi_X \; G$ to be the projection of $G$ on $X$, and $\pi_Y \; G$ to be the projection on $Y$. The subgraph $\RelGraph{\cap_Y \cdot R} \leq \completeGraph{X \times 1}$ of the composite $\cap_Y \cdot R$ for some CPM map $R : X \CPMrightarrow Y$ has exactly the same nodes as $\pi_X \; \RelGraph{R}$, but only those edges $\graphEdge{x}{x'}$ such that $\graphEdge{[x,y]}{[x',y]}$ is an edge in $\RelGraph{R}$ for some $y\in Y$. Causal maps $R$ are then those for which $\pi_X \; \RelGraph{R}$ covers all elements of $X$, and such that no $Y$-constant edges exist.
	\end{enumerate}
	\noindent As subgraphs of the complete graph $\completeGraph{X}$, the mixed states in $X$ come with a boolean lattice structure given by the subgraph (or, equivalently, subset) relation $\subseteq$, and in particular with graph union $\vee$ (or, equivalently, subset union). We can define a notion of \textbf{purity} as the partial order $\preceq$ obtained by restricting $\subseteq$ to graphs with the same node-set: this has discrete subgraphs as its minima, and cliques (pure states) as its maxima. 

	The CPM category $\CPMCategory{\fdHilbCategory}$ does not inherit the enriched structure of $\fdHilbCategory$, but it has an operation of convex combination which preserves causality. The CPM category $\CPMCategory{\fRelCategory}$, on the other hand, turns out to be closed under the superposition operation $\vee$ from $\fRelCategory$, in the form of union of graphs and preserving causality: in $\CPMCategory{\fRelCategory}$, we shall refer to this as \textbf{convex combination}. In $\CPMCategory{\fdHilbCategory}$, non-trivial convex combination of non-pure states, or of distinct pure states, is never pure. In $\CPMCategory{\fRelCategory}$, on the other hand, convex combination of non-pure states can yield a pure state.
	\begin{lemma}\label{lemma_PurityConvexCombination}
		Let $P: 1 \CPMrightarrow X$ be a pure state in $\CPMCategory{\fRelCategory}$ corresponding to a subset of $X$ with $n$ elements. For any $m = 2,...,\frac{n(n-1)}{2}$, $P$ can be expressed as a convex combination of $m$ non-pure states $(\rho_j)_{j=1,...,m}$. Furthermore, if $\rho: 1 \CPMrightarrow X$ is any mixed state, then there exists another mixed state $\rho': 1 \CPMrightarrow X$ with $\RelGraph{\rho}$ and $\RelGraph{\rho'}$ having the same node set and $\rho \vee \rho'$ a pure state. 
	\end{lemma}
	\begin{proof} 
		We prove this of the case $m=\frac{n(n-1)}{2}$, and the other cases follow easily. Let $D \preceq P$ be the state with discrete graph on the same node set of $P$, and $D \prec \rho_1,...,\rho_m$ are all the states with graphs having the same node set of $P$ and exactly 1 edge (i.e. 1 edge away from being a purity minimum in the sense defined above). Then $\vee_{j=1,...,n} \; \rho_j = P$. Furthermore, if $\rho$ is any mixed state, let $G_\rho' \leq \completeGraph{X}$ be the complement subgraph to $G_\rho$, with the same nodes as $G_\rho$ and such that an edge is in $G_\rho'$ if and only if it isn't in $G_\rho$. Then $G_\rho \vee G\rho' = \completeGraph{X}$ and hence $\rho \vee \rho'$ is pure. Also, $\rho'$ is not pure unless $G_\rho$ is a discrete subgraph.
	\end{proof} 

\newcommand{\decoherence}[1]{\operatorname{dec}(#1)}
\section{Decoherence and measurements}
	
	The definition of measurements and the treatment of the ensuing classical data is a tricky subject in quantum theory. A rigorous formalisation for $\fdHilbCategory$ can be achieved by working in $\CPMCategory{\fdHilbCategory}$ and considering a certain family of CPM maps, the decoherence maps for each classical structure on a system. Given a classical structure $\hbox{\begin{tikzpicture} [scale=1,transform shape] 

\def\deltax{0.3} 
\def\deltay{0.5} 


\node [dot, fill=\classicalStructColour] (mult) at (0,0) {};

\end{tikzpicture}}\!$ on some Hilbert space $\SpaceH$, the associated decoherence map is a CPM map $\decoherence{\hbox{\begin{tikzpicture} [scale=1,transform shape] 

\def\deltax{0.3} 
\def\deltay{0.5} 


\node [dot, fill=\classicalStructColour] (mult) at (0,0) {};

\end{tikzpicture}}\!}$ with the following property: if $\rho$ is any mixed state in $\SpaceH$, then $\decoherence{\hbox{\begin{tikzpicture} [scale=1,transform shape] 

\def\deltax{0.3} 
\def\deltay{0.5} 


\node [dot, fill=\classicalStructColour] (mult) at (0,0) {};

\end{tikzpicture}}\!} \cdot \rho$ is always a convex combination $\sum_j \, p_j\ket{j}\bra{j}$ of $\hbox{\begin{tikzpicture} [scale=1,transform shape] 

\def\deltax{0.3} 
\def\deltay{0.5} 


\node [dot, fill=\classicalStructColour] (mult) at (0,0) {};

\end{tikzpicture}}\!$-classical points. 

	Because of this property, the result of $\hbox{\begin{tikzpicture} [scale=1,transform shape] 

\def\deltax{0.3} 
\def\deltay{0.5} 


\node [dot, fill=\classicalStructColour] (mult) at (0,0) {};

\end{tikzpicture}}\!$-decoherence can always be interpreted as probabilistic $\hbox{\begin{tikzpicture} [scale=1,transform shape] 

\def\deltax{0.3} 
\def\deltay{0.5} 


\node [dot, fill=\classicalStructColour] (mult) at (0,0) {};

\end{tikzpicture}}\!$-classical data. Decoherence maps can be defined in the CPM category associated with any abstract process theory in the following way.
	\begin{definition} 
		Let $\hbox{\begin{tikzpicture} [scale=1,transform shape] 

\def\deltax{0.3} 
\def\deltay{0.5} 


\node [dot, fill=\classicalStructColour] (mult) at (0,0) {};

\end{tikzpicture}}\!$ be a classical structure on some space $X$ in some compact closed $\dagger$-SMC $\CategoryC$ Then the \textbf{$\hbox{\begin{tikzpicture} [scale=1,transform shape] 

\def\deltax{0.3} 
\def\deltay{0.5} 


\node [dot, fill=\classicalStructColour] (mult) at (0,0) {};

\end{tikzpicture}}\!$-decoherence} map $\decoherence{\hbox{\begin{tikzpicture} [scale=1,transform shape] 

\def\deltax{0.3} 
\def\deltay{0.5} 


\node [dot, fill=\classicalStructColour] (mult) at (0,0) {};

\end{tikzpicture}}\!}$ is defined to be the following causal morphism of $\CPMCategory{\CategoryC}$:
		\begin{equation}\label{Zdecoherence}
		    \hbox{\begin{tikzpicture}[node distance = 10mm]

\node (center) {};

\node (map) [dot, fill = \Zcolour] [right of = center, xshift = -4mm] {};
\node (out) [above of = map, yshift = -2mm] {};
\node (in) [below of = map, yshift = +2mm] {};

\node (mapconj) [dot, fill = \Zcolour] [left of = center, xshift = +4mm] {};
\node (outconj) [above of = mapconj, yshift = -2mm] {};
\node (inconj) [below of = mapconj, yshift = +2mm] {};

\begin{pgfonlayer}{background}
\draw[->-=.5,out=90,in=270] (in) to (map);
\draw[->-=.5,out=90,in=270] (map) to (out);
\draw[-<-=.5,out=90,in=270] (inconj) to (mapconj);
\draw[-<-=.5,out=90,in=270] (mapconj) to (outconj);
\draw[->-=.53,out=135,in=45,mark=at position 0.5 with {\arrow{>}}] (map.120) to (mapconj.60);
\end{pgfonlayer}

\end{tikzpicture}}
		\end{equation}
	\end{definition}

	\noindent In $\CPMCategory{\fdHilbCategory}$, the result of a $\hbox{\begin{tikzpicture} [scale=1,transform shape] 

\def\deltax{0.3} 
\def\deltay{0.5} 


\node [dot, fill=\classicalStructColour] (mult) at (0,0) {};

\end{tikzpicture}}\!$-decoherence map is always a convex combination $\sum_j \, p_j\ket{j}\bra{j}$ of $\hbox{\begin{tikzpicture} [scale=1,transform shape] 

\def\deltax{0.3} 
\def\deltay{0.5} 


\node [dot, fill=\classicalStructColour] (mult) at (0,0) {};

\end{tikzpicture}}\!$-classical points: as long as any further operation on the result is $\hbox{\begin{tikzpicture} [scale=1,transform shape] 

\def\deltax{0.3} 
\def\deltay{0.5} 


\node [dot, fill=\classicalStructColour] (mult) at (0,0) {};

\end{tikzpicture}}\!$-classical (i.e. only pure CPM maps coming from classical endomorphisms of $\hbox{\begin{tikzpicture} [scale=1,transform shape] 

\def\deltax{0.3} 
\def\deltay{0.5} 


\node [dot, fill=\classicalStructColour] (mult) at (0,0) {};

\end{tikzpicture}}\!$ are allowed), the entire process is equivalent to the state $\sum_j \, p_j\ket{j}$ being acted upon in $\fdHilbCategory$ by the same classical endomorphisms. This leads to a so-called quantum-classical formalism, where the decoherence map from Diagram \ref{Zdecoherence} is replaced with the following map, and it is understood that any operation following it\footnote{Where categorical composition is read bottom to top.} will be $
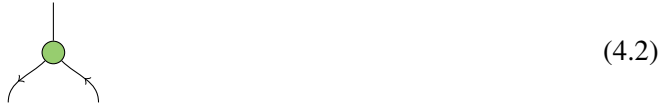
\hbox{\begin{tikzpicture} [scale=1,transform shape] 

\def\deltax{0.3} 
\def\deltay{0.5} 


\node [dot, fill=\classicalStructColour] (mult) at (0,0) {};

\end{tikzpicture}}\!$-classical:
	\begin{equation}\label{ZdecoherenceQC}
	    \hbox{\begin{tikzpicture}[node distance = 10mm]

\node (center) [dot, fill = \Zcolour] {};
\node (out) [above of = center, yshift = -2mm] {};

\node (map) [right of = center, xshift = -4mm] {};
\node (in) [below of = map, yshift = +2mm] {};

\node (mapconj) [left of = center, xshift = +4mm] {};
\node (inconj) [below of = mapconj, yshift = +2mm] {};

\begin{pgfonlayer}{background}
\draw[->-=.5,out=90,in=315] (in) to (center);
\draw[-,out=90,in=270] (center) to (out);
\draw[-<-=.5,out=90,in=225] (inconj) to (center);
\end{pgfonlayer}

\end{tikzpicture}}
	\end{equation}

	\noindent Unfortunately, in $\CPMCategory{\fRelCategory}$ things are not this simple: the $\hbox{\begin{tikzpicture} [scale=1,transform shape] 

\def\deltax{0.3} 
\def\deltay{0.5} 


\node [dot, fill=\classicalStructColour] (mult) at (0,0) {};

\end{tikzpicture}}\!$-decoherence map applied to a mixed state does not in general return a convex combination of $\hbox{\begin{tikzpicture} [scale=1,transform shape] 

\def\deltax{0.3} 
\def\deltay{0.5} 


\node [dot, fill=\classicalStructColour] (mult) at (0,0) {};

\end{tikzpicture}}\!$-classical points. The short-cut summarised by the map in Diagram \ref{ZdecoherenceQC} does not work, and we have to work entirely in the CPM category if we want to make proper sense of measurements and decoherence in $\fRelCategory$. 

	\begin{definition}
		Let $\hbox{\begin{tikzpicture} [scale=1,transform shape] 

\def\deltax{0.3} 
\def\deltay{0.5} 


\node [dot, fill=\classicalStructColour] (mult) at (0,0) {};

\end{tikzpicture}}\!$ be a classical structure in $\fRelCategory$ on some set $X$, associated with the abelian groupoid $\oplus_{\lambda\in  \Lambda}\; G_\lambda$, and let $d_\lambda \in  G_\lambda$. We define the \textbf{orbit subgraph} for $d_\lambda$ as the subgraph of $\completeGraph{X}$ with nodes $\graphNode{g_\lambda}$ for all $g_\lambda \in  G_\lambda$ and edges $\graphEdge{g_\lambda}{(g_\lambda + d_\lambda)}$ for all $g_\lambda \in  G_\lambda$. Intuitively, the orbit subgraph traces the orbit of $d_\lambda$ under the right regular action $g_\lambda \mapsto g_\lambda + d_\lambda$.
	\end{definition}

	\begin{theorem}\label{thm_CPMReldecoherence}
		Let $\hbox{\begin{tikzpicture} [scale=1,transform shape] 

\def\deltax{0.3} 
\def\deltay{0.5} 


\node [dot, fill=\classicalStructColour] (mult) at (0,0) {};

\end{tikzpicture}}\!$ be a classical structure in $\fRelCategory$ on some set $X$, associated with the abelian groupoid $\oplus_{\lambda\in  \Lambda}\; G_\lambda$. Let $\sigma = \decoherence{\hbox{\begin{tikzpicture} [scale=1,transform shape] 

\def\deltax{0.3} 
\def\deltay{0.5} 


\node [dot, fill=\classicalStructColour] (mult) at (0,0) {};

\end{tikzpicture}}\!} \cdot \rho$ be the mixed state resulting from decohering some state $\rho$. Then the subgraph $\RelGraph{\sigma}\leq \completeGraph{X}$ is the union of the orbit subgraphs for $d_\lambda = g_\lambda - g'_\lambda$, for all $d_\lambda$ such that: 
		\begin{equation}
			\exists \, g_\lambda,g'_\lambda \in G_\lambda \text{ s.t. } \graphEdge{g_\lambda}{g'_\lambda} \text{ appears in } \RelGraph{\rho}
		\end{equation}
		An example graph for the abelian groupoid $\integersMod{2} \oplus \integersMod{3}$ on a 5-element set is given in \ref{ZdecoherenceGraph} below.
	\end{theorem}

	\begin{proof} We see the decoherence map as a morphism $X \times X \rightarrow X \times X$ in $\fRelCategory$, and try to evaluate the composition $\bra{h'_{\tilde{\gamma}}}\bra{h_\gamma}\decoherence{\hbox{\begin{tikzpicture} [scale=1,transform shape] 

\def\deltax{0.3} 
\def\deltay{0.5} 


\node [dot, fill=\classicalStructColour] (mult) at (0,0) {};

\end{tikzpicture}}\!}\ket{g'_{\tilde{\lambda}}}\ket{g_\lambda}$. Consider the $\fRelCategory$ morphism $X \rightarrow X$ obtained by composing the state $\ket{g_\lambda}$ and effect $\bra{h'_{\tilde{\gamma}}}$ to the bottom right and top left wires of figure \ref{Zdecoherence} respectively:
		\begin{enumerate}
			\item[a.] on the RHS, we have $\ZcomultSym \cdot \ket{g_\lambda} = \vee_{h_\lambda\in  G_\lambda} \ket{g_\lambda - h_\lambda} \ket{h_\lambda}$
			\item[b.] we now have $\ket{g_\lambda - h_\lambda}$ going on the central wire and $\ket{h_\lambda}$ going up the top right wire  
			\item[c.] on the LHS, we have $\ZmultSym \cdot (\ket{g_\lambda-h_\lambda} \ket{h'_{\tilde{\gamma}}}) = \delta_{\lambda \tilde{\gamma}} \ket{g_\lambda-h_\lambda + h'_\lambda}$
			\item[d.] we now have $\delta_{\lambda \tilde{\gamma}} \ket{g_\lambda - h_\lambda + h'_\lambda}$ going down the bottom left wire
		\end{enumerate} 
		We conclude that the decoherence map, seen as a morphism in $\fRelCategory$, yields the scalar $1$, when applied to $\ket{g'_{\tilde{\lambda}}} \ket{g_\lambda}$ and evaluated against effects $\bra{h'_{\tilde{\gamma}}} \bra{h_\gamma}$, if and only if $\lambda = \tilde{\lambda} = \gamma = \tilde{\gamma}$ and $g'_\lambda = g_\lambda -h_\lambda + h'_\lambda$. Thus the decoherence map, seen as a CPM morphism, sends a node $\graphNode{g_\lambda}$ to the node set $\{h_\lambda \in  G_\lambda\}$, and an edge $\graphEdge{g_\lambda}{g'_{\tilde{\lambda}}}$ to the orbit subgraph for $d_\lambda := g_\lambda - g'_\lambda$ if $\lambda = \tilde{\lambda}$, and to no edge otherwise.
	\end{proof}

	
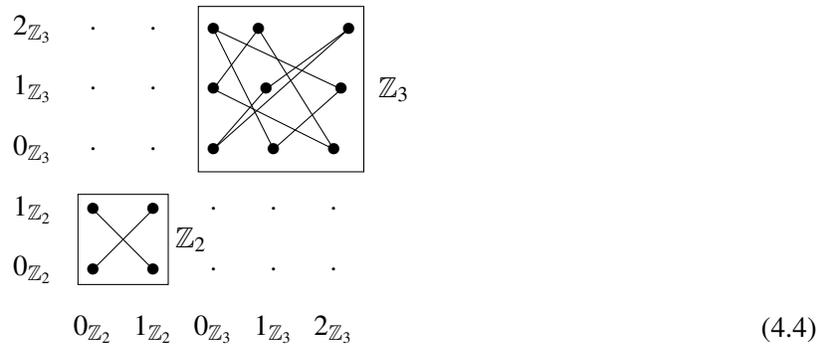
\begin{equation}\label{ZdecoherenceGraph}
	    \hbox{\begin{tikzpicture}[node distance = 8mm]


\node (anchor) {};

\node (Xa0) [above of = anchor] {$0_{\integersMod{2}}$};
\node (Xa1) [above of = Xa0] {$1_{\integersMod{2}}$};
\node (Xb0) [above of = Xa1] {$0_{\integersMod{3}}$};
\node (Xb1) [above of = Xb0] {$1_{\integersMod{3}}$};
\node (Xb2) [above of = Xb1] {$2_{\integersMod{3}}$};

\node (Ya0) [right of = anchor] {$0_{\integersMod{2}}$};
\node (Ya1) [right of = Ya0] {$1_{\integersMod{2}}$};
\node (Yb0) [right of = Ya1] {$0_{\integersMod{3}}$};
\node (Yb1) [right of = Yb0] {$1_{\integersMod{3}}$};
\node (Yb2) [right of = Yb1] {$2_{\integersMod{3}}$};

\node(a0a0) [right of = Xa0] {$\bullet$};
\node(a0a1) [right of = a0a0] {$\bullet$};
\node(a0b0) [right of = a0a1] {$\cdot$};
\node(a0b1) [right of = a0b0, xshift = 0mm] {$\cdot$};
\node(a0b2) [right of = a0b1, xshift = 0mm] {$\cdot$};

\node(a1a0) [right of = Xa1] {$\bullet$};
\node(a1a1) [right of = a1a0] {$\bullet$};
\node(a1b0) [right of = a1a1] {$\cdot$};
\node(a1b1) [right of = a1b0, xshift = 0mm] {$\cdot$};
\node(a1b2) [right of = a1b1, xshift = 0mm] {$\cdot$};

\node(b0a0) [right of = Xb0] {$\cdot$};
\node(b0a1) [right of = b0a0] {$\cdot$};
\node(b0b0) [right of = b0a1] {$\bullet$};
\node(b0b1) [right of = b0b0, xshift = 0mm] {$\bullet$};
\node(b0b2) [right of = b0b1, xshift = 0mm] {$\bullet$};

\node(b1a0) [right of = Xb1] {$\cdot$};
\node(b1a1) [right of = b1a0] {$\cdot$};
\node(b1b0) [right of = b1a1] {$\bullet$};
\node(b1b1) [right of = b1b0, xshift = -1mm] {$\bullet$};
\node(b1b2) [right of = b1b1, xshift = +2mm] {$\bullet$};

\node(b2a0) [right of = Xb2] {$\cdot$};
\node(b2a1) [right of = b2a0] {$\cdot$};
\node(b2b0) [right of = b2a1] {$\bullet$};
\node(b2b1) [right of = b2b0, xshift = -2mm] {$\bullet$};
\node(b2b2) [right of = b2b1, xshift = +4mm] {$\bullet$};

\begin{pgfonlayer}{background}
\draw[-] (a0a0.center) to (a1a1.center);
\draw[-] (a0a1.center) to (a1a0.center);
\draw[-] (b0b0.center) to (b1b1.center);  
\draw[-] (b1b1.center) to (b2b2.center);  
\draw[-] (b0b0.center) to (b2b2.center);  
\draw[-] (b0b1.center) to (b1b2.center);  
\draw[-] (b1b2.center) to (b2b0.center);  
\draw[-] (b2b0.center) to (b0b1.center);  
\draw[-] (b1b0.center) to (b2b1.center);  
\draw[-] (b2b1.center) to (b0b2.center);  
\draw[-] (b0b2.center) to (b1b0.center);  
\end{pgfonlayer}

\node (a) [box, fill=none, right of = Xa0, xshift = 4mm, yshift = 4mm, minimum size = 12mm] {};
\node (alabel) [right of = Xa0, xshift = 13mm, yshift = 4mm] {$\integersMod{2}$};

\node (b) [box, fill=none, right of = Xb1, xshift = 25mm, yshift = 0mm, minimum size = 22mm] {};
\node (blabel) [right of = Xb1, xshift = 40mm, yshift = 0mm] {$\integersMod{3}$};

\end{tikzpicture}}
	\end{equation}

	\noindent As a consequence of Theorem \ref{thm_CPMReldecoherence}, $\sigma = \decoherence{\hbox{\begin{tikzpicture} [scale=1,transform shape] 

\def\deltax{0.3} 
\def\deltay{0.5} 


\node [dot, fill=\classicalStructColour] (mult) at (0,0) {};

\end{tikzpicture}}\!} \cdot \rho$ can be written as:
	\begin{equation}
		\sigma = \bigvee_{\lambda\in \Lambda'} \tau_\lambda  \;\;\;\text{ for some }\Lambda' \subseteq \Lambda \text{ and some } \tau_\lambda \preceq \ket{G_\lambda}\bra{G_\lambda}
	\end{equation} 
	The mixed state $\sigma$ is a convex combination of $\hbox{\begin{tikzpicture} [scale=1,transform shape] 

\def\deltax{0.3} 
\def\deltay{0.5} 


\node [dot, fill=\classicalStructColour] (mult) at (0,0) {};

\end{tikzpicture}}\!$-classical points if only if $\tau_\lambda = \ket{G_\lambda}\bra{G_\lambda}$ for all $\lambda \in  \Lambda'$, which is not in general the case: this is enough to call into question the assumption that decoherence maps provide, in general, a suitable quantum-classical interface (an assumption which, for example, is built into the conceptual backbone of the CP* construction \cite{CQM-CategoriesQuantumClassicalChannels}). 

	However, this may in principle just be an issue with the specific form \ref{ZdecoherenceQC} chosen for decoherence maps, and there could be some different morphism which works for classical structures with not enough classical points: as it turns out, there isn't one.
	\begin{lemma}
		Let $\hbox{\begin{tikzpicture} [scale=1,transform shape] 

\def\deltax{0.3} 
\def\deltay{0.5} 


\node [dot, fill=\classicalStructColour] (mult) at (0,0) {};

\end{tikzpicture}}\!$ be a classical structure in $\fRelCategory$ on some set $Z$, associated with the abelian groupoid $\oplus_{\lambda \in \Lambda} \; G_\lambda$. If there is a CPM morphism $d: Z \CPMrightarrow Z$ which satisfies the following three properties, then $\hbox{\begin{tikzpicture} [scale=1,transform shape] 

\def\deltax{0.3} 
\def\deltay{0.5} 


\node [dot, fill=\classicalStructColour] (mult) at (0,0) {};

\end{tikzpicture}}\!$ is the discrete structure and $d = \decoherence{\hbox{\begin{tikzpicture} [scale=1,transform shape] 

\def\deltax{0.3} 
\def\deltay{0.5} 


\node [dot, fill=\classicalStructColour] (mult) at (0,0) {};

\end{tikzpicture}}\!}$:
		\begin{enumerate}
			\item[(a)] $d( \ket{G_\lambda}\bra{G_\lambda}) = \ket{G_\lambda}\bra{G_\lambda}$ for all $\lambda \in \Lambda$
			\item[(b)] For all CPM states $\rho: 1 \CPMrightarrow Z$, $d(\rho) = \vee_{\lambda \in \Lambda'_\rho} \ket{G_\lambda}\bra{\lambda}$ for some subset $\Lambda'_\rho \subseteq \Lambda$
			\item[(c)] $d$ is a demolition measurement (see Definition \ref{def_Measurement} below)
		\end{enumerate}
	\end{lemma}
	\begin{proof}
		Assume that one such $d$ exists, then the following must hold:
		\begin{enumerate}
			\item[(i)] by requirement (a), $\RelGraph{d}$ is a subgraph of $\bigvee_{\lambda \in \Lambda} \RelGraph{\ket{G_\lambda}\bra{G_\lambda}}$
			\item[(ii)] in particular, by point (i) and requirement (b), if $g_\lambda \in G_\lambda$ then $d(\ket{g_\lambda}\bra{g_\lambda}) = G_\lambda$ 
			\item[(iii)] by point (ii), $\RelGraph{d}$ must have edges in the form $\graphEdge{[g_\lambda,h_\lambda]}{[g_\lambda,h'_\lambda]}$ for all $g_\lambda,h_\lambda,h'_\lambda \in G_\lambda$
			\item[(iv)] by requirement (c), $d^\dagger$ must be causal, and in particular if $\graphEdge{[g_\lambda,h_\lambda]}{[g_\lambda,h'_\lambda]}$ is an edge of $\RelGraph{d}$ then $h_\lambda = h'_\lambda$
		\end{enumerate}
		By points (iii) and (iv) above, we conclude that all $G_\lambda$ are singletons, and therefore $\hbox{\begin{tikzpicture} [scale=1,transform shape] 

\def\deltax{0.3} 
\def\deltay{0.5} 


\node [dot, fill=\classicalStructColour] (mult) at (0,0) {};

\end{tikzpicture}}\!$ is the discrete structure. By requirement (c), $d$ must be causal, and combining this with point (i) above we obtain that $d$ must be the decoherence map for the discrete structure $\hbox{\begin{tikzpicture} [scale=1,transform shape] 

\def\deltax{0.3} 
\def\deltay{0.5} 


\node [dot, fill=\classicalStructColour] (mult) at (0,0) {};

\end{tikzpicture}}\!$.  
	\end{proof}
	
	\noindent We conclude that the identification of decohered states in $\CPMCategory{\fRelCategory}$ as possibilistic mixtures of classical data is not sound, and should be avoided unless further measures are in place (e.g. a suitable equivalence relation on CPM states and morphisms). 
	
	Having understood decoherence, it is finally time to introduce measurements in $\CPMCategory{\fRelCategory}$ and tackle the fundamental question of locality. We set aside the issues with convex mixing and take a more traditional approach, using demolition measurements and evaluating against classical points (the equivalence relation we were talking about) to obtain a possibilistic empirical model. We will discuss other options for future work in the conclusions. 

	Definition \ref{def_Measurement} introduces non-demolition and demolition measurements in $\CPMCategory{\fRelCategory}$, in accordance with the framework laid out in \cite{CQM-QuantumMeasuNoSums} and the upcoming \cite{CQM-QCSnotes}. Theorem \ref{thm_MeasurementDecoherence} shows that the same results of a $\hbox{\begin{tikzpicture} [scale=1,transform shape] 

\def\deltax{0.3} 
\def\deltay{0.5} 


\node [dot, fill=\classicalStructColour] (mult) at (0,0) {};

\end{tikzpicture}}\!$-demolition measurement (with evaluation against $\hbox{\begin{tikzpicture} [scale=1,transform shape] 

\def\deltax{0.3} 
\def\deltay{0.5} 


\node [dot, fill=\classicalStructColour] (mult) at (0,0) {};

\end{tikzpicture}}\!$-classical points) can be obtained by decoherence in some other structure $\hbox{\begin{tikzpicture} [scale=1,transform shape] 

\def\deltax{0.3} 
\def\deltay{0.5} 


\node [dot, fill=\groupStructColour] (mult) at (0,0) {};

\end{tikzpicture}}\!$, followed by some classical manipulation of the results. 

	Definition \ref{def_empiricalModel} defines the desired class of empirical models for $\fRelCategory$: by virtue of Theorem \ref{thm_MeasurementDecoherence}, we only need to consider measurements given directly by decoherence maps. Theorem \ref{thm_LHV} finally shows that every measurement scenario has a local hidden variable model, settling once and for all that $\fRelCategory$ is local.

	\begin{definition}\label{def_Measurement}
		Let $\hbox{\begin{tikzpicture} [scale=1,transform shape] 

\def\deltax{0.3} 
\def\deltay{0.5} 


\node [dot, fill=\classicalStructColour] (mult) at (0,0) {};

\end{tikzpicture}}\!$ be a classical structure in $\fRelCategory$ on some set $Z$, associated with the abelian groupoid $\oplus_{\lambda\in  \Lambda}\; G_\lambda$. A $\hbox{\begin{tikzpicture} [scale=1,transform shape] 

\def\deltax{0.3} 
\def\deltay{0.5} 


\node [dot, fill=\classicalStructColour] (mult) at (0,0) {};

\end{tikzpicture}}\!$-valued \textbf{non-demolition measurement} on some set $X$ is a causal CPM map $M: X \CPMrightarrow X \times Z$ in the following form, and satisfying the  idempotence and self-adjointness properties below: 
		
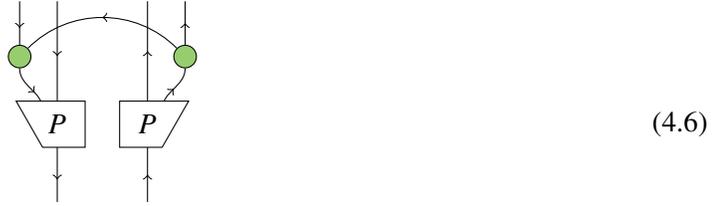
\begin{equation}\label{eqn_NonDemolitionMeasurement}
			\hbox{\begin{tikzpicture}[node distance = 6mm]


\node (center) {};

\node [map] (map) [right of = center] {$P$};
\node [mapconj] (mapconj) [left of = center] {$P$};

\node (in) [below of = map, yshift = -3mm] {};
\node (inconj) [below of = mapconj, yshift = -3mm] {};

\node [dot, fill=\Zcolour] (dot) [above of = map, yshift = 3mm, xshift = +5mm] {};
\node [dot, fill=\Zcolour] (dotconj) [above of = mapconj, yshift = 3mm, xshift = -5mm] {};

\node (out) [above of = map, yshift = 9mm] {};
\node (outconj) [above of = mapconj, yshift = 9mm] {};

\node (cout) [above of = dot] {};
\node (coutconj) [above of = dotconj] {};

\begin{pgfonlayer}{background}
\draw[->-=.5] (in.270) to (map);
\draw[-<-=.5] (inconj.270) to (mapconj);
\draw[->-=.5] (map) to (out.90);
\draw[-<-=.5] (mapconj) to (outconj.90);
\draw[->-=.5] [out=90,in=270] (map.45) to (dot.270);
\draw[-<-=.5] [out=90,in=270] (mapconj.135) to (dotconj.270);
\draw[->-=.5] [out=90,in=270] (dot) to (cout.90);
\draw[-<-=.5] [out=90,in=270] (dotconj) to (coutconj.90);
\draw[->-=.5] [out=135,in=45] (dot) to (dotconj);
\end{pgfonlayer}

\end{tikzpicture}}
		\end{equation}
		Let $M_{\lambda} := \left(\id{X} \times \rho_{G_\lambda}^{\dagger}\right) \cdot M$ be $M$ evaluated against the pure state $\rho_{G_\lambda} := \ket{G_\lambda}\ket{G_\lambda}$ in $\CPMCategory{\fRelCategory}$ corresponding to the classical points $G_\lambda$ of $\hbox{\begin{tikzpicture} [scale=1,transform shape] 

\def\deltax{0.3} 
\def\deltay{0.5} 


\node [dot, fill=\classicalStructColour] (mult) at (0,0) {};

\end{tikzpicture}}\!$. Then:
		\begin{enumerate}
			\item[(a.)] $M$ is \textbf{idempotent} if it satisfies the following equation:
			
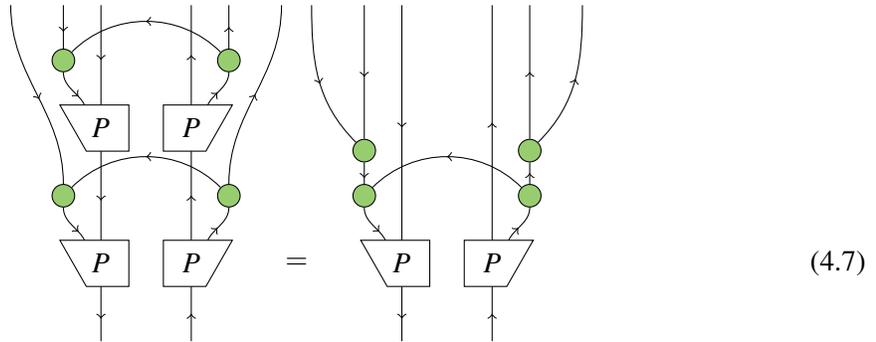
\begin{equation}\label{eqn_NDMeasIdempotence}
				\hbox{\begin{tikzpicture}[node distance = 6mm]


\node (center) {};

\node [map] (map) [right of = center] {$P$};
\node [mapconj] (mapconj) [left of = center] {$P$};

\node (in) [below of = map, yshift = -3mm] {};
\node (inconj) [below of = mapconj, yshift = -3mm] {};

\node [dot, fill=\Zcolour] (dot) [above of = map, yshift = 3mm, xshift = +5mm] {};
\node [dot, fill=\Zcolour] (dotconj) [above of = mapconj, yshift = 3mm, xshift = -5mm] {};

\node (cout) [above of = dot, yshift = 18mm, xshift = +7mm] {};
\node (coutconj) [above of = dotconj, yshift = 18mm, xshift = -7mm] {};

\node [map] (map2) [above of = map, yshift = 12mm] {$P$};
\node [mapconj] (mapconj2) [above of = mapconj, yshift = 12mm]{$P$};

\node (in2) [below of = map2, yshift = -3mm] {};
\node (inconj2) [below of = mapconj2, yshift = -3mm] {};

\node [dot, fill=\Zcolour] (dot2) [above of = map2, yshift = 3mm, xshift = +5mm] {};
\node [dot, fill=\Zcolour] (dotconj2) [above of = mapconj2, yshift = 3mm, xshift = -5mm] {};

\node (cout2) [above of = dot2] {};
\node (coutconj2) [above of = dotconj2] {};

\node (out) [above of = map2, yshift = 9mm] {};
\node (outconj) [above of = mapconj2, yshift = 9mm] {};

\begin{pgfonlayer}{background}
\draw[->-=.5] (in.270) to (map);
\draw[-<-=.5] (inconj.270) to (mapconj);

\draw[->-=.5] [out=90,in=270] (map.45) to (dot.270);
\draw[-<-=.5] [out=90,in=270] (mapconj.135) to (dotconj.270);
\draw[->-=.5] [out=90,in=270] (dot) to (cout.90);
\draw[-<-=.5] [out=90,in=270] (dotconj) to (coutconj.90);
\draw[->-=.5] [out=135,in=45] (dot) to (dotconj);

\draw[->-=.5] (map) to (map2);
\draw[-<-=.5] (mapconj) to (mapconj2);

\draw[->-=.5] [out=90,in=270] (map2.45) to (dot2.270);
\draw[-<-=.5] [out=90,in=270] (mapconj2.135) to (dotconj2.270);
\draw[->-=.5] [out=90,in=270] (dot2) to (cout2.90);
\draw[-<-=.5] [out=90,in=270] (dotconj2) to (coutconj2.90);
\draw[->-=.5] [out=135,in=45] (dot2) to (dotconj2);

\draw[->-=.5] (map2) to (out.90);
\draw[-<-=.5] (mapconj2) to (outconj.90);
\end{pgfonlayer}

\node (equal) [right of = center, xshift = 14mm] {$=$};

\node (center) [right of = equal, xshift = 14mm]{};

\node [map] (map) [right of = center] {$P$};
\node [mapconj] (mapconj) [left of = center] {$P$};

\node (in) [below of = map, yshift = -3mm] {};
\node (inconj) [below of = mapconj, yshift = -3mm] {};

\node [dot, fill=\Zcolour] (dot) [above of = map, yshift = 3mm, xshift = +5mm] {};
\node [dot, fill=\Zcolour] (dotconj) [above of = mapconj, yshift = 3mm, xshift = -5mm] {};

\node (out) [above of = map, yshift = 27mm] {};
\node (outconj) [above of = mapconj, yshift = 27mm] {};

\node [dot, fill=\Zcolour] (diag) [above of = dot] {};
\node [dot, fill=\Zcolour] (diagconj) [above of = dotconj] {};

\node (cout) [above of = diag, yshift = 12mm, xshift = +7mm] {};
\node (coutconj) [above of = diagconj, yshift = 12mm, xshift = -7mm] {};

\node (cout2) [above of = diag, yshift = 12mm] {};
\node (coutconj2) [above of = diagconj, yshift = 12mm] {};

\begin{pgfonlayer}{background}
\draw[->-=.5] (in.270) to (map);
\draw[-<-=.5] (inconj.270) to (mapconj);
\draw[->-=.5] (map) to (out.90);
\draw[-<-=.5] (mapconj) to (outconj.90);
\draw[->-=.5] [out=90,in=270] (map.45) to (dot.270);
\draw[-<-=.5] [out=90,in=270] (mapconj.135) to (dotconj.270);
\draw[->-=.5] [out=90,in=270] (dot) to (diag);
\draw[-<-=.5] [out=90,in=270] (dotconj) to (diagconj);
\draw[->-=.5] [out=135,in=45] (dot) to (dotconj);

\draw[->-=.5] [out=45,in=270] (diag) to (cout.90);
\draw[-<-=.5] [out=135,in=270] (diagconj) to (coutconj.90);
\draw[->-=.5] [out=90,in=270] (diag) to (cout2.90);
\draw[-<-=.5] [out=90,in=270] (diagconj) to (coutconj2.90);
\end{pgfonlayer}

\end{tikzpicture}}
			\end{equation}
			In particular, for all $\lambda \in  \Lambda$ we have $M_\lambda \cdot M_\lambda = M_\lambda$
			\item[(b.)] $M$ is \textbf{self-adjoint} if it satisfies the following equation:
			
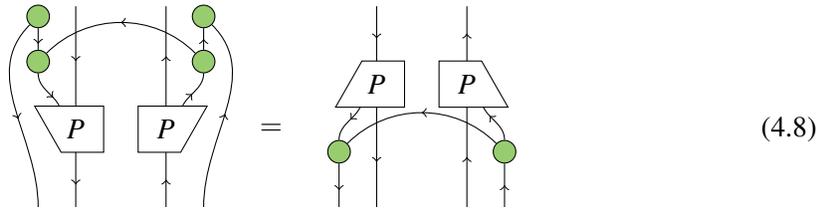
\begin{equation}\label{eqn_NDMeasSelfadjointness}
				\hbox{\begin{tikzpicture}[node distance = 6mm]


\node (center) {};

\node [map] (map) [right of = center] {$P$};
\node [mapconj] (mapconj) [left of = center] {$P$};

\node (in) [below of = map, yshift = -3mm] {};
\node (inconj) [below of = mapconj, yshift = -3mm] {};

\node [dot, fill=\Zcolour] (dot) [above of = map, yshift = 3mm, xshift = +5mm] {};
\node [dot, fill=\Zcolour] (dotconj) [above of = mapconj, yshift = 3mm, xshift = -5mm] {};

\node (out) [above of = map, yshift = 9mm] {};
\node (outconj) [above of = mapconj, yshift = 9mm] {};

\node [dot, fill=\Zcolour] (cout) [above of = dot] {};
\node [dot, fill=\Zcolour] (coutconj) [above of = dotconj] {};

\node (cin) [below of = cout, yshift = -18mm] {};
\node (cinconj) [below of = coutconj, yshift = -18mm] {};

\begin{pgfonlayer}{background}
\draw[->-=.5] (in.270) to (map);
\draw[-<-=.5] (inconj.270) to (mapconj);
\draw[->-=.5] (map) to (out.90);
\draw[-<-=.5] (mapconj) to (outconj.90);
\draw[->-=.5] [out=90,in=270] (map.45) to (dot.270);
\draw[-<-=.5] [out=90,in=270] (mapconj.135) to (dotconj.270);
\draw[->-=.5] [out=90,in=270] (dot) to (cout);
\draw[-<-=.5] [out=90,in=270] (dotconj) to (coutconj);
\draw[->-=.5] [out=135,in=45] (dot) to (dotconj);

\draw[->-=.5] [out=90,in=315] (cin.270) to (cout);
\draw[-<-=.5] [out=90,in=225] (cinconj.270) to (coutconj);
\end{pgfonlayer}

\node (equals) [right of = center, xshift = 14mm]{$=$};

\node (center) [right of = equals, xshift = 14mm]{};

\node [mapdag] (map) [right of = center, yshift = 6mm] {$P$};
\node [maptrans] (mapconj) [left of = center, yshift = 6mm] {$P$};

\node (in) [below of = map, yshift = -9mm] {};
\node (inconj) [below of = mapconj, yshift = -9mm] {};

\node [dot, fill=\Zcolour] (dot) [below of = map, yshift = -3mm, xshift = +5mm] {};
\node [dot, fill=\Zcolour] (dotconj) [below of = mapconj, yshift = -3mm, xshift = -5mm] {};

\node (out) [above of = map, yshift = 3mm] {};
\node (outconj) [above of = mapconj, yshift = 3mm] {};

\node (cin) [below of = map, xshift = 5mm, yshift = -9mm] {};
\node (cinconj) [below of = mapconj, xshift = -5mm, yshift = -9mm] {};

\begin{pgfonlayer}{background}
\draw[->-=.5] (in.270) to (map);
\draw[-<-=.5] (inconj.270) to (mapconj);
\draw[->-=.5] (map) to (out.90);
\draw[-<-=.5] (mapconj) to (outconj.90);
\draw[-<-=.5] [out=270,in=90] (map.315) to (dot);
\draw[->-=.5] [out=270,in=90] (mapconj.225) to (dotconj);

\draw[->-=.5] [out=135,in=45] (dot) to (dotconj);

\draw[->-=.5] [out=90,in=270] (cin.270) to (dot);
\draw[-<-=.5] [out=90,in=270] (cinconj.270) to (dotconj);
\end{pgfonlayer}

\end{tikzpicture}}
			\end{equation}
			In particular, for all $\lambda \in  \Lambda$ we have $M_\lambda^{\dagger} = M_\lambda$
		\end{enumerate}
		A $\hbox{\begin{tikzpicture} [scale=1,transform shape] 

\def\deltax{0.3} 
\def\deltay{0.5} 


\node [dot, fill=\classicalStructColour] (mult) at (0,0) {};

\end{tikzpicture}}\!$-valued \textbf{demolition measurement} on $X$ is a CPM map $\bar{M}: X \CPMrightarrow Z$ in the form $\bar{M} = (\cap_X\times \id{Z}) \cdot M$ for some $\hbox{\begin{tikzpicture} [scale=1,transform shape] 

\def\deltax{0.3} 
\def\deltay{0.5} 


\node [dot, fill=\classicalStructColour] (mult) at (0,0) {};

\end{tikzpicture}}\!$-valued non-demolition measurement $M$. Thus demolition measurements are exactly the CPM maps obtained by tracing out the set $X$ in a non-demolition measurement on $X$.
	\end{definition}

	\noindent In particular the decoherence maps are demolition measurements: it turns out that, in $\fRelCategory$, they are the only measurements we ever need.
	\begin{theorem}\label{thm_MeasurementDecoherence}
		Let $\hbox{\begin{tikzpicture} [scale=1,transform shape] 

\def\deltax{0.3} 
\def\deltay{0.5} 


\node [dot, fill=\classicalStructColour] (mult) at (0,0) {};

\end{tikzpicture}}\!$ be a classical structure in $\fRelCategory$ on some set $Z$, associated with the abelian groupoid $\oplus_{\lambda\in  \Lambda}\; G_\lambda$. Let $\bar{M}: X \CPMrightarrow Z$ be a $\hbox{\begin{tikzpicture} [scale=1,transform shape] 

\def\deltax{0.3} 
\def\deltay{0.5} 


\node [dot, fill=\classicalStructColour] (mult) at (0,0) {};

\end{tikzpicture}}\!$-valued demolition measurement. Let $\bar{M}_\lambda := \rho_{G_\lambda} \cdot \bar{M} : X \CPMrightarrow 1$ be $\bar{M}$ evaluated against the $\hbox{\begin{tikzpicture} [scale=1,transform shape] 

\def\deltax{0.3} 
\def\deltay{0.5} 


\node [dot, fill=\classicalStructColour] (mult) at (0,0) {};

\end{tikzpicture}}\!$-classical point $\ket{G_\lambda}$ of $\hbox{\begin{tikzpicture} [scale=1,transform shape] 

\def\deltax{0.3} 
\def\deltay{0.5} 


\node [dot, fill=\classicalStructColour] (mult) at (0,0) {};

\end{tikzpicture}}\!$. Then there exist:\\
		(i) a classical structure $\hbox{\begin{tikzpicture} [scale=1,transform shape] 

\def\deltax{0.3} 
\def\deltay{0.5} 


\node [dot, fill=\groupStructColour] (mult) at (0,0) {};

\end{tikzpicture}}\!$ on $X$ (for some abelian groupoid $\oplus_{\gamma\in  \Gamma}\; H_\gamma$)\\
		(ii) an endomorphism of the classical structures corresponding to some $f: \Gamma \stackrel{\SetCategory}{\rightarrow} \Lambda$\\
		such that the following holds, where $D_\gamma: X \CPMrightarrow 1$ is the $\hbox{\begin{tikzpicture} [scale=1,transform shape] 

\def\deltax{0.3} 
\def\deltay{0.5} 


\node [dot, fill=\groupStructColour] (mult) at (0,0) {};

\end{tikzpicture}}\!$-decoherence map evaluated against $\ket{H_\gamma}$:
		\begin{equation}
			\bar{M}_\lambda = \bigvee_{\gamma \text{ s.t. } f(\gamma) = \lambda} \; D_\gamma
		\end{equation}
	\end{theorem}
	\begin{proof}
		The measurement $M$ is a causal CPM map: thus the $P$ map in Diagram \ref{eqn_measurement} is an isometry, and hence by Lemma \ref{lemma_IsometriesUnitarisFRel} it is a classical map in the following form, with $J_{s(x)}$ classical points of some classical structure on $X \times Y$ (given by some abelian groupoid $\oplus_{\delta\in  \Delta}\; J_\delta$) and $s: X \stackrel{\SetCategory}{\rightarrow} \Delta$ a classical injection:
		\begin{equation}
			P = \bigvee_{x \in  X} \;\ket{J_{s(x)}}\bra{x}
		\end{equation}  
		
		\noindent Until further notice we work in $\fRelCategory$ (i.e. with pure maps only). The subsets $(G_\lambda)_{\lambda\in \Lambda}$ are disjoint, and so are the subsets $(J_{s(x)})_{x\in X}$: thus the family of states $\ket{D_{\lambda,x}}:= \left(\id{X} \times \bra{G_\lambda} \right) \cdot \ket{J_{s(x)}}$ is also composed of disjoint subsets of $X$. 

		Note that the CPM map $M_\lambda$ is pure, since the central wire of the $\hbox{\begin{tikzpicture} [scale=1,transform shape] 

\def\deltax{0.3} 
\def\deltay{0.5} 


\node [dot, fill=\classicalStructColour] (mult) at (0,0) {};

\end{tikzpicture}}\!$-decoherence map disappears once the decoherence is evaluated against a pure $\hbox{\begin{tikzpicture} [scale=1,transform shape] 

\def\deltax{0.3} 
\def\deltay{0.5} 


\node [dot, fill=\classicalStructColour] (mult) at (0,0) {};

\end{tikzpicture}}\!$-classical state: we can continue working in $\fRelCategory$, with $P_\lambda:= \left(\id{X} \times \bra{G_\lambda} \right) \cdot P$ in place of $M_\lambda$ (which lives in $\CPMCategory{\fRelCategory}$). Given the $D_{\lambda,x}$ defined above, the map $P_\lambda$ takes the following form:
		\begin{equation}
			P_\lambda = \bigvee_{x\in X} \; \ket{D_{\lambda,x}}\bra{x}
		\end{equation}

		\noindent Now $M$ is idempotent/self-adjoint if and only if all $M_\lambda$ are idempotent/self-adjoint, if and only if all $P_\lambda$ are idempotent/self-adjoint (because the $M_\lambda$ are pure maps). Let $R_\lambda$ be the following relation on $X$:
		\begin{equation}
			x R_\lambda y \text{ if and only if } y \in  D_{\lambda,x}
		\end{equation}
		
		\noindent Self-adjointness of $P_\lambda$ is equivalent to $R_\lambda$ being symmetric, while idempotence of $P_\lambda$ is equivalent to $R_\lambda$ being transitive; since all $x\in X$ appear in $R_\lambda$, then by symmetry and transitivity $R_\lambda$ is also reflexive. 
		Thus $R_\lambda$ is an equivalence relation on $X$ for all $\lambda$. 

		Now we go back to $\CPMCategory{\fRelCategory}$: the subgraph $\RelGraph{\bar{M}_\lambda} \leq \completeGraph{X}$ associated with the CPM map $\bar{M}_\lambda : X \CPMrightarrow 1$ has edges $\graphEdge{x}{x'}$ for all $x,x' \in X $ such that $x R_\lambda y$ and $x' R_\lambda y$ for some $y \in  X$. 
		But $R_\lambda$ is an equivalence relation, so $\RelGraph{\bar{M}_\lambda}$ is the union of the cliques on the equivalence classes of $R_\lambda$. 

		For each $\lambda \in  \Lambda$, let $(H_{(\lambda,j)})_{j=1,...,n_\lambda}$ be any family of groups, each with element sets one of the $n_\lambda$ equivalence classes of $R_\lambda$. 
		Define $\Gamma := \suchthat{(\lambda,j)}{\lambda \in  \Lambda, j=1,...,n_\lambda}$, and let $\hbox{\begin{tikzpicture} [scale=1,transform shape] 

\def\deltax{0.3} 
\def\deltay{0.5} 


\node [dot, fill=\groupStructColour] (mult) at (0,0) {};

\end{tikzpicture}}\!$ be the classical structure on $X$ associated with the abelian groupoid $\oplus_{\gamma\in  \Gamma}\; H_\gamma$. Then each state $(\bar{M}_\lambda)^\dagger$ is a convex combination of $\hbox{\begin{tikzpicture} [scale=1,transform shape] 

\def\deltax{0.3} 
\def\deltay{0.5} 


\node [dot, fill=\groupStructColour] (mult) at (0,0) {};

\end{tikzpicture}}\!$-classical points. 
		
		Thus the results of the demolition measurement $\bar{M}$ can be reproduced by the $\hbox{\begin{tikzpicture} [scale=1,transform shape] 

\def\deltax{0.3} 
\def\deltay{0.5} 


\node [dot, fill=\groupStructColour] (mult) at (0,0) {};

\end{tikzpicture}}\!$-decoherence map, i.e. by testing against $\hbox{\begin{tikzpicture} [scale=1,transform shape] 

\def\deltax{0.3} 
\def\deltay{0.5} 


\node [dot, fill=\groupStructColour] (mult) at (0,0) {};

\end{tikzpicture}}\!$-classical points and then applying the function $f = (\lambda,j) \mapsto \lambda$ to reconstruct the $\bar{M}$ measurements result in terms of $\hbox{\begin{tikzpicture} [scale=1,transform shape] 

\def\deltax{0.3} 
\def\deltay{0.5} 


\node [dot, fill=\classicalStructColour] (mult) at (0,0) {};

\end{tikzpicture}}\!$-classical points.
	\end{proof}

	\noindent We have shown that the only demolition measurements we really need are the decoherence maps. But applying a decoherence map and then testing against a classical state is the same as testing directly against the classical state, so we can give the following, simpler definition of an empirical model, where mixed states are directly tested against classical points, with no demolition measurements in between.

	\begin{definition}\label{def_empiricalModel}
		Let $\rho$ be a mixed state in $X_1 \times ... \times X_N$. For each $j=1,...,N$, let $(\hbox{\begin{tikzpicture} [scale=1,transform shape] 

\def\deltax{0.3} 
\def\deltay{0.5} 


\node [dot, fill=\groupStructColour] (mult) at (0,0) {};

\end{tikzpicture}}\!_j^m)_{m=1,...,M}$ be a family of classical structures on $X_j$. Let $(\Lambda_j^m)_{jm}$ be sets indexing the classical points of the classical structures (without loss of generality, disjoint for different classical structures). Let $\Phi^m (\lambda_1^m,...,\lambda_N^m): \Lambda_1^m \times...\times \Lambda_N^m \rightarrow \{\bot,\top\}$ be a boolean function defined by the scalar obtained evaluating $\rho$ against the separable pure state $\rho_{\lambda_1^m} \times ... \times \rho_{\lambda_N^m}$ associated with the family of classical points $\ket{G_{\lambda_j^m}}_{j=1,...,N}$, as shown in Equation \ref{eqn_EmpiricalModel}. Then $(\Phi^m)_m$ is a \textbf{possibilistic empirical model} (based on mixed state $\rho$).
		\begin{equation}\label{eqn_EmpiricalModel}
			\hbox{\begin{tikzpicture}[node distance = 10mm]

\node (center) {};

\node (label) [left of = center, xshift = -50mm,yshift = -3mm] {$\Phi^m(\lambda_1^m,...,\lambda_N^m)$};

\node (label) [left of = center,xshift = -32mm,yshift = -3mm] {$ := $};

\node (lambda1center) [above of = center, xshift = -20mm, yshift = -5mm] {};
\node (lambda1) [kpointdag, right of = lambda1center, xshift = -4mm] {$G_{\lambda_1^m}$};
\node (lambda1in) [below of = lambda1, yshift = +2mm] {};
\node (lambda1star) [kpointtrans, left of = lambda1center, xshift = +4mm] {$G_{\lambda_1^m}$};
\node (lambda1starin) [below of = lambda1star, yshift = +2mm] {};

\node (dots) [above of = center, yshift = -5mm] {$...$};

\node (lambdaNcenter) [above of = center, xshift = +20mm, yshift = -5mm] {};
\node (lambdaN) [kpointdag, right of = lambdaNcenter, xshift = -4mm] {$G_{\lambda_N^m}$};
\node (lambdaNin) [below of = lambdaN, yshift = +2mm] {};
\node (lambdaNstar) [kpointtrans, left of = lambdaNcenter, xshift = +4mm] {$G_{\lambda_N^m}$};
\node (lambdaNstarin) [below of = lambdaNstar, yshift = +2mm] {};

\node (rho) [box, minimum width = 70mm, below of = center, xshift = 0mm, yshift = +5mm] {$\rho$};

\begin{pgfonlayer}{background}

\draw[-<-=.5,out=90,in=270] (lambda1starin) to (lambda1star);
\draw[->-=.5,out=90,in=270] (lambda1in) to (lambda1);

\draw[-<-=.5,out=90,in=270] (lambdaNstarin) to (lambdaNstar);
\draw[->-=.5,out=90,in=270] (lambdaNin) to (lambdaN);

\end{pgfonlayer}

\end{tikzpicture}}
		\end{equation}
	\end{definition}

	\noindent This definition of (possibilistic) empirical model squares with that given in \cite{NLC-SheafSeminal}, in the following way:
	\begin{enumerate}
		\item[1.] The commutative semiring is that of the booleans $(\{\bot,\top\},\vee,\bot,\wedge,\top)$.
		\item[2.] The set of measurements is $\suchthat{\decoherence{\hbox{\begin{tikzpicture} [scale=1,transform shape] 

\def\deltax{0.3} 
\def\deltay{0.5} 


\node [dot, fill=\groupStructColour] (mult) at (0,0) {};

\end{tikzpicture}}\!}_j^m}{m=1,...,M \text{ and }j=1,...,N}$. As a consequence of Theorem \ref{thm_MeasurementDecoherence}, this is general enough to capture all measurements in $\CPMCategory{\fRelCategory}$.\footnote{The classical function involved in Theorem \ref{thm_MeasurementDecoherence} irrelevant when it comes to non-locality.}\footnote{To be more precise, all measurements according to the framework laid out in \cite{CQM-QuantumMeasuNoSums} and the upcoming \cite{CQM-QCSnotes}, which is restricted to \textit{commutative} special $\dagger$-Frobenius Algebras.}
		\item[3.] The measurement contexts take the form $C_m := \suchthat{\decoherence{\hbox{\begin{tikzpicture} [scale=1,transform shape] 

\def\deltax{0.3} 
\def\deltay{0.5} 


\node [dot, fill=\groupStructColour] (mult) at (0,0) {};

\end{tikzpicture}}\!}_j^m}{j=1,...,N}$, for $m = 1,...,M$.
		\item[4.] The sheaf of events is defined by $\mathcal{E}(C_m) := \prod_{j=1,...,N}\Lambda_j^m$, with measurement-dependent outcomes.
		\item[5.] The empirical model is the family $(\Phi^m:\mathcal{E}(C_m) \rightarrow \{\bot,\top\})_{m}$ of boolean-valued distributions.
	\end{enumerate}

	\noindent Under the correspondence above, the local hidden variable we construct in Theorem \ref{thm_LHV} below take the form of those defined in \cite{NLC-SheafSeminal} (further details can be found in the Appendix). In this sense, our last result should be interpreted as stating that $\fRelCategory$ is \textbf{local}: every empirical model, obtained by considering demolition measurements as per Definition \ref{def_Measurement}, admits a local hidden variable.

	\begin{theorem}\label{thm_LHV}
		Every possibilistic empirical model $(\Phi^m)_m$ in $\CPMCategory{\fRelCategory}$ constructed as in Definition \ref{def_empiricalModel} admits a local hidden variable $\nu$. This is shown in Figure \ref{eqn_LHV}, and is obtained as follows:\\
		(i) the mixed state $\rho$ underlying the empirical model is decohered in the discrete structures on $X_1,...,X_N$;
		(ii) the discrete classical data is appropriately copied to obtain a local hidden variable $\nu$.
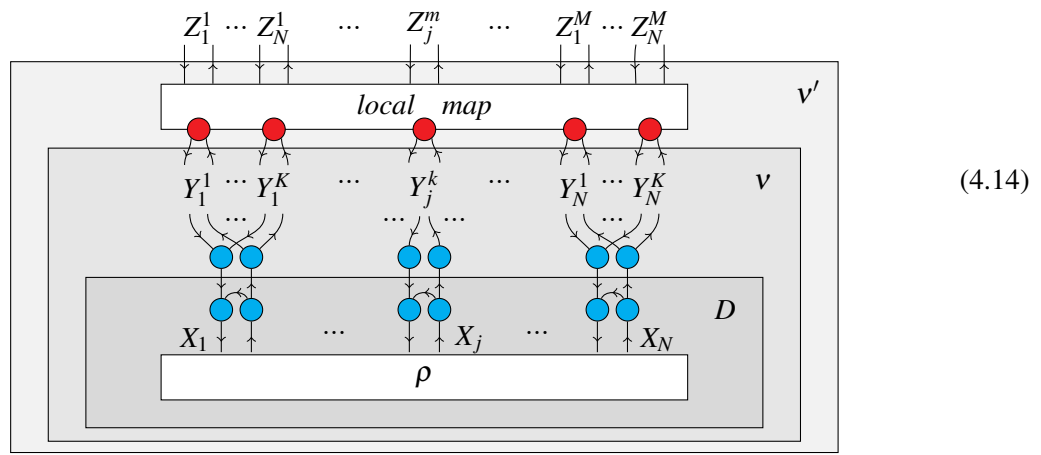
		\begin{equation}\label{eqn_LHV}
			\hbox{\begin{tikzpicture}[node distance = 10mm]

\node (center) {};

\begin{pgfonlayer}{background}

\definecolor{light-gray}{gray}{0.80}
\definecolor{very-light-gray}{gray}{0.85}
\definecolor{super-light-gray}{gray}{0.90}	
\definecolor{super-duper-light-gray}{gray}{0.95}	

\node (nu) [box, fill = super-duper-light-gray, below of = center, yshift = +0mm, minimum width = 110mm, minimum height = 52mm] {};
\node (nuprime) [box, fill = super-light-gray, below of = center, yshift = -5mm, minimum width = 100mm, minimum height = 39mm] {};
\node (D) [box, fill = very-light-gray, below of = center, yshift = -12.7mm, minimum width = 90mm, minimum height = 20mm] {};

\end{pgfonlayer}

\node (nulabel) [above of = center, xshift = 51mm,yshift = 2mm] {$\nu'$};

\node (nulabel) [below of = center, xshift = 45mm,yshift = +10mm] {$\nu$};

\node (nulabel) [below of = center, xshift = 40mm,yshift = -7mm] {$D$};

\node (localmap)[box, minimum width = 70mm, above of = center] {local $\;\;$ map};

\node (Z11) [above of = localmap, inner sep = 0mm, xshift = -30mm, yshift = +0.5mm] {$Z_1^1$};
\node (Z11in) [inner sep = 3.2mm, below of = Z11, yshift = -0.7mm] {};

\node (Z1dots) [above of = localmap, xshift = -25mm, yshift = +0.5mm] {$...$};

\node (Z1N) [above of = localmap, inner sep = 0mm, xshift = -20mm, yshift = +0.5mm] {$Z_N^1$};
\node (Z1Nin) [inner sep = 3.2mm, below of = Z1N, yshift = -0.7mm] {};

\node (Z1Zjdots) [above of = localmap, xshift = -10mm, yshift = +0.5mm] {$...$};

\node (Zjm) [above of = localmap, inner sep = 0mm, xshift = 0mm, yshift = +0.5mm] {$Z_j^m$};
\node (Zjmin) [inner sep = 3.2mm, below of = Zjm, yshift = -0.7mm] {};

\node (ZjZNdots) [above of = localmap, xshift = +10mm, yshift = +0.5mm] {$...$};

\node (ZN1) [above of = localmap, inner sep = 0mm, xshift = +20mm, yshift = +0.5mm] {$Z_1^M$};
\node (ZN1in) [inner sep = 3.2mm, below of = ZN1, yshift = -0.7mm] {};

\node (ZNdots) [above of = localmap, xshift = +25mm, yshift = +0.5mm] {$...$};

\node (ZMN) [above of = localmap, inner sep = 0mm, xshift = +30mm, yshift = +0.5mm] {$Z_N^M$};
\node (ZMNin) [inner sep = 3.2mm, below of = ZMN, yshift = -0.7mm] {};

\node (Y11) [below of = localmap, inner sep = 1mm, xshift = -30mm, yshift = -1mm] {$Y_1^1$};
\node (Y11in) [dot, fill = \Xcolour, above of = Y11, yshift = -2mm] {};

\node (Y1dots) [below of = localmap, xshift = -25mm, yshift = 0mm] {$...$};

\node (Y1M) [below of = localmap, inner sep = 1mm, xshift = -20mm, yshift = -1mm] {$Y_1^K$};
\node (Y1Min) [dot, fill = \Xcolour, above of = Y1M, yshift = -2mm] {};

\node (Y1Zjdots) [below of = localmap, xshift = -10mm, yshift = 0mm] {$...$};

\node (Yjm) [below of = localmap, inner sep = 1mm, xshift = 0mm, yshift = -1mm] {$Y_j^k$};
\node (Yjmin) [dot, fill = \Xcolour, above of = Yjm, yshift = -2mm] {};

\node (YjYNdots) [below of = localmap, xshift = +10mm, yshift = 0mm] {$...$};

\node (YN1) [below of = localmap, inner sep = 1mm, xshift = +20mm, yshift = -1mm] {$Y_N^1$};
\node (YN1in) [dot, fill = \Xcolour, above of = YN1, yshift = -2mm] {};

\node (YNdots) [below of = localmap, xshift = +25mm, yshift = 0mm] {$...$};

\node (YNM) [below of = localmap, inner sep = 1mm, xshift = +30mm, yshift = -1mm] {$Y_N^K$};
\node (YNMin) [dot, fill = \Xcolour, above of = YNM, yshift = -2mm] {};

\node (X1multl) [dot, fill = \Zaltcolour, below of = localmap, xshift = -27mm, yshift = -10mm] {};
\node (X1multr) [dot, fill = \Zaltcolour, below of = localmap, xshift = -23mm, yshift = -10mm] {};
\node (X1dots) [below of = localmap, xshift = -25mm, yshift = -5mm] {$...$};

\node (Xjmultl) [dot, fill = \Zaltcolour, below of = localmap, xshift = -2mm, yshift = -10mm] {};
\node (Xjmultr) [dot, fill = \Zaltcolour, below of = localmap, xshift = +2mm, yshift = -10mm] {};
\node (Xjdotsl) [below of = localmap, xshift = -4mm, yshift = -5mm] {$...$};
\node (Xjdotsl) [below of = localmap, xshift = +4mm, yshift = -5mm] {$...$};

\node (XNmultl) [dot, fill = \Zaltcolour, below of = localmap, xshift = +23mm, yshift = -10mm] {};
\node (XNmultr) [dot, fill = \Zaltcolour, below of = localmap, xshift = +27mm, yshift = -10mm] {};
\node (XNdots) [below of = localmap, xshift = +25mm, yshift = -5mm] {$...$};

\node (X1decl) [dot, fill = \Zaltcolour, below of = X1multl, xshift = 0mm, yshift = +3mm] {};
\node (X1declbot) [ below of = X1decl, xshift = 0mm, yshift = +3mm] {};
\node (X1decr) [dot, fill = \Zaltcolour, below of = X1multr, xshift = 0mm, yshift = +3mm] {};
\node (X1decrbot) [ below of = X1decr, xshift = 0mm, yshift = +3mm] {};
\node (X1label) [ below of = X1decl, xshift = -3.5mm, yshift = +6mm] {$X_1$};

\node (Xjdecl) [dot, fill = \Zaltcolour, below of = Xjmultl, xshift = 0mm, yshift = +3mm] {};
\node (Xjdeclbot) [ below of = Xjdecl, xshift = 0mm, yshift = +3mm] {};
\node (Xjdecr) [dot, fill = \Zaltcolour, below of = Xjmultr, xshift = 0mm, yshift = +3mm] {};
\node (Xjdecrbot) [ below of = Xjdecr, xshift = 0mm, yshift = +3mm] {};
\node (Xjlabel) [ below of = Xjdecr, xshift = 4mm, yshift = +6mm] {$X_j$};

\node (XNdecl) [dot, fill = \Zaltcolour, below of = XNmultl, xshift = 0mm, yshift = +3mm] {};
\node (XNdeclbot) [ below of = XNdecl, xshift = 0mm, yshift = +3mm] {};
\node (XNdecr) [dot, fill = \Zaltcolour, below of = XNmultr, xshift = 0mm, yshift = +3mm] {};
\node (XNdecrbot) [ below of = XNdecr, xshift = 0mm, yshift = +3mm] {};
\node (XNlabel) [ below of = XNdecr, xshift = 4mm, yshift = +6mm] {$X_N$};

\node (X1Xjdots) [ below of = center, xshift = -12mm, yshift = -10mm] {$...$};
\node (XjXNdots) [ below of = center, xshift = +15mm, yshift = -10mm] {$...$};

\node (rho) [box, minimum width = 70mm, below of = localmap, xshift = 0mm, yshift = -26mm] {$\rho$};

\begin{pgfonlayer}{background}
\draw[-<-=.5,out=90,in=270] (Z11in.120) to (Z11.230);
\draw[->-=.5,out=90,in=270] (Z11in.60) to (Z11.310);
\draw[-<-=.5,out=90,in=270] (Z1Nin.120) to (Z1N.230);
\draw[->-=.5,out=90,in=270] (Z1Nin.60) to (Z1N.310);
\draw[-<-=.5,out=90,in=270] (Zjmin.120) to (Zjm.230);
\draw[->-=.5,out=90,in=270] (Zjmin.60) to (Zjm.310);
\draw[-<-=.5,out=90,in=270] (ZN1in.120) to (ZN1.230);
\draw[->-=.5,out=90,in=270] (ZN1in.60) to (ZN1.310);
\draw[-<-=.5,out=90,in=260] (ZMNin.120) to (ZMN.230);
\draw[->-=.5,out=90,in=270] (ZMNin.60) to (ZMN.310);

\draw[-<-=.5,out=90,in=270] (Y11.120) to (Y11in.230);
\draw[->-=.5,out=90,in=270] (Y11.60) to (Y11in.310);
\draw[-<-=.5,out=90,in=270] (Y1M.120) to (Y1Min.230);
\draw[->-=.5,out=90,in=270] (Y1M.60) to (Y1Min.310);
\draw[-<-=.5,out=90,in=270] (Yjm.120) to (Yjmin.230);
\draw[->-=.5,out=90,in=270] (Yjm.60) to (Yjmin.310);
\draw[-<-=.5,out=90,in=270] (YN1.120) to (YN1in.230);
\draw[->-=.5,out=90,in=270] (YN1.60) to (YN1in.310);
\draw[-<-=.5,out=90,in=270] (YNM.120) to (YNMin.230);
\draw[->-=.5,out=90,in=270] (YNM.60) to (YNMin.310);

\draw[-<-=.5,out=135,in=270] (X1multl) to (Y11.250);
\draw[->-=.5,out=135,in=270] (X1multr) to (Y11.290);
\draw[-<-=.5,out=45,in=270] (X1multl) to (Y1M.250);
\draw[->-=.5,out=45,in=270] (X1multr) to (Y1M.290);

\draw[-<-=.5,out=90,in=270] (Xjmultl) to (Yjm.260);
\draw[->-=.5,out=90,in=270] (Xjmultr) to (Yjm.280);

\draw[-<-=.5,out=135,in=270] (XNmultl) to (YN1.250);
\draw[->-=.5,out=135,in=270] (XNmultr) to (YN1.290);
\draw[-<-=.5,out=45,in=270] (XNmultl) to (YNM.250);
\draw[->-=.5,out=45,in=270] (XNmultr) to (YNM.290);

\draw[-<-=.5,out=90,in=270] (X1decl) to (X1multl);
\draw[->-=.5,out=90,in=270] (X1decr) to (X1multr);
\draw[->-=.5,out=110,in=70] (X1decr) to (X1decl);
\draw[-<-=.5,out=90,in=270] (X1declbot) to (X1decl);
\draw[->-=.5,out=90,in=270] (X1decrbot) to (X1decr);

\draw[-<-=.5,out=90,in=270] (Xjdecl) to (Xjmultl);
\draw[->-=.5,out=90,in=270] (Xjdecr) to (Xjmultr);
\draw[->-=.5,out=110,in=70] (Xjdecr) to (Xjdecl);
\draw[-<-=.5,out=90,in=270] (Xjdeclbot) to (Xjdecl);
\draw[->-=.5,out=90,in=270] (Xjdecrbot) to (Xjdecr);

\draw[-<-=.5,out=90,in=270] (XNdecl) to (XNmultl);
\draw[->-=.5,out=90,in=270] (XNdecr) to (XNmultr);
\draw[->-=.5,out=110,in=70] (XNdecr) to (XNdecl);
\draw[-<-=.5,out=90,in=270] (XNdeclbot) to (XNdecl);
\draw[->-=.5,out=90,in=270] (XNdecrbot) to (XNdecr);

\end{pgfonlayer}

\end{tikzpicture}}
		\end{equation}
	\end{theorem}

	\begin{proof} The proof hinges on the following, somewhat puzzling fact: the purity of a mixed state cannot be observed in $\CPMCategory{\fRelCategory}$. To be more precise, let $\rho,\tau$ be mixed states and $\sigma$ be any other mixed state such that $\tau \preceq \rho$: then $\sigma^\dagger \cdot \rho = \sigma^\dagger \cdot \tau$ (the proof of this can be found below). A picture of a \textit{purity lower-set} in a 12-element set can be found below in \ref{eqn_CPMstatesPurity}: none of the 8 states presented there can be distinguished from the discrete state.

	
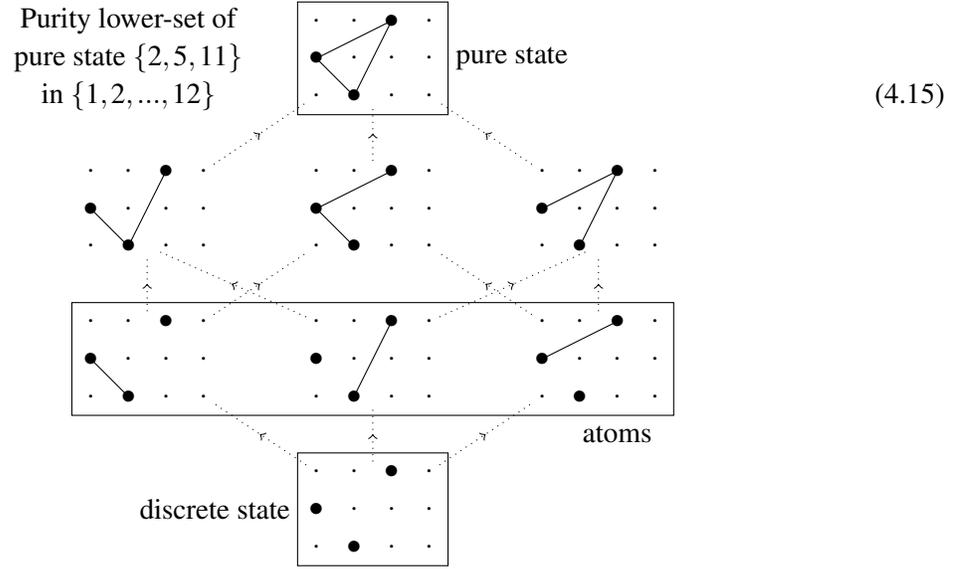
\begin{equation}\label{eqn_CPMstatesPurity}
		\hbox{\begin{tikzpicture}[node distance = 5mm]

\node (x0) {};
\node (x11) [right of = x0] {$\cdot$};
\node (x12) [right of = x11] {$\bullet$};
\node (x13) [right of = x12] {$\cdot$};
\node (x14) [right of = x13] {$\cdot$};
\node (x21) [above of = x11] {$\bullet$};
\node (x22) [right of = x21] {$\cdot$};
\node (x23) [right of = x22] {$\cdot$};
\node (x24) [right of = x23] {$\cdot$};
\node (x31) [above of = x21] {$\cdot$};
\node (x32) [right of = x31] {$\cdot$};
\node (x33) [right of = x32] {$\bullet$};
\node (x34) [right of = x33] {$\cdot$};

\node (text) [left of = x0, xshift = -15mm, yshift = 10mm] {Purity lower-set of};
\node (text) [left of = x0, xshift = -15mm, yshift = 5mm] {pure state $\{2,5,11\}$};
\node (text) [left of = x0, xshift = -15mm, yshift = 0mm] {in $\{1,2,...,12\}$};

\begin{pgfonlayer}{background}
\node (box) [box, fill=none, minimum width = 20mm, minimum height = 15mm, right of = x22, xshift = -2.5mm] {};
\end{pgfonlayer}

\node (text) [right of = x24, xshift = 6mm] {pure state};

\begin{pgfonlayer}{background}
\draw[-] (x12.center) to (x21.center);
\draw[-] (x12.center) to (x33.center);
\draw[-] (x21.center) to (x33.center);
\end{pgfonlayer}

\node (x0anchorl) [right of = x0] {};
\node (x0anchorc) [right of =x0, xshift = 7.5mm] {};
\node (x0anchorr) [right of = x0, xshift = 15mm] {};

\node (x0r1) [below of = x0, xshift = -30mm, yshift = -15mm]{};
\node (x11) [right of = x0r1] {$\cdot$};
\node (x12) [right of = x11] {$\bullet$};
\node (x13) [right of = x12] {$\cdot$};
\node (x14) [right of = x13] {$\cdot$};
\node (x21) [above of = x11] {$\bullet$};
\node (x22) [right of = x21] {$\cdot$};
\node (x23) [right of = x22] {$\cdot$};
\node (x24) [right of = x23] {$\cdot$};
\node (x31) [above of = x21] {$\cdot$};
\node (x32) [right of = x31] {$\cdot$};
\node (x33) [right of = x32] {$\bullet$};
\node (x34) [right of = x33] {$\cdot$};

\begin{pgfonlayer}{background}
\draw[-] (x12.center) to (x21.center);
\draw[-] (x12.center) to (x33.center);
\end{pgfonlayer}

\node (x0r1anchorlt) [right of = x0r1,yshift = 10mm] {};
\node (x0r1anchorct) [right of =x0r1, xshift = 7.5mm,yshift = 10mm] {};
\node (x0r1anchorrt) [right of = x0r1, xshift = 15mm,yshift = 10mm] {};
\begin{pgfonlayer}{background}
\draw[dotted->-=.5] (x0r1anchorrt) to (x0anchorl);
\end{pgfonlayer}

\node (x0r2) [below of = x0, xshift = 0mm, yshift = -15mm]{};
\node (x11) [right of = x0r2] {$\cdot$};
\node (x12) [right of = x11] {$\bullet$};
\node (x13) [right of = x12] {$\cdot$};
\node (x14) [right of = x13] {$\cdot$};
\node (x21) [above of = x11] {$\bullet$};
\node (x22) [right of = x21] {$\cdot$};
\node (x23) [right of = x22] {$\cdot$};
\node (x24) [right of = x23] {$\cdot$};
\node (x31) [above of = x21] {$\cdot$};
\node (x32) [right of = x31] {$\cdot$};
\node (x33) [right of = x32] {$\bullet$};
\node (x34) [right of = x33] {$\cdot$};

\begin{pgfonlayer}{background}
\draw[-] (x12.center) to (x21.center);
\draw[-] (x21.center) to (x33.center);
\end{pgfonlayer}

\node (x0r2anchorlt) [right of = x0r2,yshift = 10mm] {};
\node (x0r2anchorct) [right of =x0r2, xshift = 7.5mm,yshift = 10mm] {};
\node (x0r2anchorrt) [right of = x0r2, xshift = 15mm,yshift = 10mm] {};
\begin{pgfonlayer}{background}
\draw[dotted->-=.5] (x0r2anchorct) to (x0anchorc);
\end{pgfonlayer}

\node (x0r3) [below of = x0, xshift = +30mm, yshift = -15mm]{};
\node (x11) [right of = x0r3] {$\cdot$};
\node (x12) [right of = x11] {$\bullet$};
\node (x13) [right of = x12] {$\cdot$};
\node (x14) [right of = x13] {$\cdot$};
\node (x21) [above of = x11] {$\bullet$};
\node (x22) [right of = x21] {$\cdot$};
\node (x23) [right of = x22] {$\cdot$};
\node (x24) [right of = x23] {$\cdot$};
\node (x31) [above of = x21] {$\cdot$};
\node (x32) [right of = x31] {$\cdot$};
\node (x33) [right of = x32] {$\bullet$};
\node (x34) [right of = x33] {$\cdot$};

\begin{pgfonlayer}{background}
\draw[-] (x12.center) to (x33.center);
\draw[-] (x21.center) to (x33.center);
\end{pgfonlayer}

\node (x0r3anchorlt) [right of = x0r3,yshift = 10mm] {};
\node (x0r3anchorct) [right of =x0r3, xshift = 7.5mm,yshift = 10mm] {};
\node (x0r3anchorrt) [right of = x0r3, xshift = 15mm,yshift = 10mm] {};
\begin{pgfonlayer}{background}
\draw[dotted->-=.5] (x0r3anchorlt) to (x0anchorr);
\end{pgfonlayer}

\node (x0r1anchorl) [right of = x0r1] {};
\node (x0r1anchorc) [right of =x0r1, xshift = 7.5mm] {};
\node (x0r1anchorr) [right of = x0r1, xshift = 15mm] {};

\node (x0r2anchorl) [right of = x0r2] {};
\node (x0r2anchorc) [right of =x0r2, xshift = 7.5mm] {};
\node (x0r2anchorr) [right of = x0r2, xshift = 15mm] {};

\node (x0r3anchorl) [right of = x0r3] {};
\node (x0r3anchorc) [right of =x0r3, xshift = 7.5mm] {};
\node (x0r3anchorr) [right of = x0r3, xshift = 15mm] {};

\node (x0r1) [below of = x0, xshift = -30mm, yshift = -35mm]{};
\node (x11) [right of = x0r1] {$\cdot$};
\node (x12) [right of = x11] {$\bullet$};
\node (x13) [right of = x12] {$\cdot$};
\node (x14) [right of = x13] {$\cdot$};
\node (x21) [above of = x11] {$\bullet$};
\node (x22) [right of = x21] {$\cdot$};
\node (x23) [right of = x22] {$\cdot$};
\node (x24) [right of = x23] {$\cdot$};
\node (x31) [above of = x21] {$\cdot$};
\node (x32) [right of = x31] {$\cdot$};
\node (x33) [right of = x32] {$\bullet$};
\node (x34) [right of = x33] {$\cdot$};

\begin{pgfonlayer}{background}
\draw[-] (x12.center) to (x21.center);
\end{pgfonlayer}

\node (x0r1anchorlt) [right of = x0r1,yshift = 10mm] {};
\node (x0r1anchorct) [right of =x0r1, xshift = 7.5mm,yshift = 10mm] {};
\node (x0r1anchorrt) [right of = x0r1, xshift = 15mm,yshift = 10mm] {};
\begin{pgfonlayer}{background}
\draw[dotted->-=.5] (x0r1anchorct) to (x0r1anchorc);
\draw[dotted->-=.5] (x0r1anchorrt) to (x0r2anchorl);
\end{pgfonlayer}

\node (x0r2) [below of = x0, xshift = 0mm, yshift = -35mm]{};
\node (x11) [right of = x0r2] {$\cdot$};
\node (x12) [right of = x11] {$\bullet$};
\node (x13) [right of = x12] {$\cdot$};
\node (x14) [right of = x13] {$\cdot$};
\node (x21) [above of = x11] {$\bullet$};
\node (x22) [right of = x21] {$\cdot$};
\node (x23) [right of = x22] {$\cdot$};
\node (x24) [right of = x23] {$\cdot$};
\node (x31) [above of = x21] {$\cdot$};
\node (x32) [right of = x31] {$\cdot$};
\node (x33) [right of = x32] {$\bullet$};
\node (x34) [right of = x33] {$\cdot$};

\begin{pgfonlayer}{background}
\draw[-] (x12.center) to (x33.center);
\end{pgfonlayer}

\node (x0r2anchorlt) [right of = x0r2,yshift = 10mm] {};
\node (x0r2anchorct) [right of =x0r2, xshift = 7.5mm,yshift = 10mm] {};
\node (x0r2anchorrt) [right of = x0r2, xshift = 15mm,yshift = 10mm] {};
\begin{pgfonlayer}{background}
\draw[dotted->-=.5] (x0r2anchorlt) to (x0r1anchorc);
\draw[dotted->-=.5] (x0r2anchorrt) to (x0r3anchorc);
\end{pgfonlayer}

\begin{pgfonlayer}{background}
\node (box) [box, fill=none, minimum width = 80mm, minimum height = 15mm, right of = x22, xshift = -2.5mm] {};
\end{pgfonlayer}

\node (x0r3) [below of = x0, xshift = +30mm, yshift = -35mm]{};
\node (x11) [right of = x0r3] {$\cdot$};
\node (x12) [right of = x11] {$\bullet$};
\node (x13) [right of = x12] {$\cdot$};
\node (x14) [right of = x13] {$\cdot$};
\node (x21) [above of = x11] {$\bullet$};
\node (x22) [right of = x21] {$\cdot$};
\node (x23) [right of = x22] {$\cdot$};
\node (x24) [right of = x23] {$\cdot$};
\node (x31) [above of = x21] {$\cdot$};
\node (x32) [right of = x31] {$\cdot$};
\node (x33) [right of = x32] {$\bullet$};
\node (x34) [right of = x33] {$\cdot$};

\node (text) [below of = x13] {atoms};

\begin{pgfonlayer}{background}
\draw[-] (x21.center) to (x33.center);
\end{pgfonlayer}

\node (x0r3anchorlt) [right of = x0r3,yshift = 10mm] {};
\node (x0r3anchorct) [right of =x0r3, xshift = 7.5mm,yshift = 10mm] {};
\node (x0r3anchorrt) [right of = x0r3, xshift = 15mm,yshift = 10mm] {};
\begin{pgfonlayer}{background}
\draw[dotted->-=.5] (x0r3anchorlt) to (x0r2anchorr);
\draw[dotted->-=.5] (x0r3anchorct) to (x0r3anchorc);
\end{pgfonlayer}

\node (x0r1anchorl) [right of = x0r1] {};
\node (x0r1anchorc) [right of =x0r1, xshift = 7.5mm] {};
\node (x0r1anchorr) [right of = x0r1, xshift = 15mm] {};

\node (x0r2anchorl) [right of = x0r2] {};
\node (x0r2anchorc) [right of =x0r2, xshift = 7.5mm] {};
\node (x0r2anchorr) [right of = x0r2, xshift = 15mm] {};

\node (x0r3anchorl) [right of = x0r3] {};
\node (x0r3anchorc) [right of =x0r3, xshift = 7.5mm] {};
\node (x0r3anchorr) [right of = x0r3, xshift = 15mm] {};


\node (x0r2) [below of = x0, xshift = 0mm, yshift = -55mm]{};
\node (x11) [right of = x0r2] {$\cdot$};
\node (x12) [right of = x11] {$\bullet$};
\node (x13) [right of = x12] {$\cdot$};
\node (x14) [right of = x13] {$\cdot$};
\node (x21) [above of = x11] {$\bullet$};
\node (x22) [right of = x21] {$\cdot$};
\node (x23) [right of = x22] {$\cdot$};
\node (x24) [right of = x23] {$\cdot$};
\node (x31) [above of = x21] {$\cdot$};
\node (x32) [right of = x31] {$\cdot$};
\node (x33) [right of = x32] {$\bullet$};
\node (x34) [right of = x33] {$\cdot$};

\node (text) [left of = x21, xshift = -8.5mm] {discrete state};

\begin{pgfonlayer}{background}
\node (box) [box, fill=none, minimum width = 20mm, minimum height = 15mm, right of = x22, xshift = -2.5mm] {};
\end{pgfonlayer}

\begin{pgfonlayer}{background}
\end{pgfonlayer}

\node (x0r2anchorlt) [right of = x0r2,yshift = 10mm] {};
\node (x0r2anchorct) [right of =x0r2, xshift = 7.5mm,yshift = 10mm] {};
\node (x0r2anchorrt) [right of = x0r2, xshift = 15mm,yshift = 10mm] {};
\begin{pgfonlayer}{background}
\draw[dotted->-=.5] (x0r2anchorlt) to (x0r1anchorr);
\draw[dotted->-=.5] (x0r2anchorct) to (x0r2anchorc);
\draw[dotted->-=.5] (x0r2anchorrt) to (x0r3anchorl);
\end{pgfonlayer}

\end{tikzpicture}}
	\end{equation}

	Now consider an empirical model $(\Phi^m)_m$, based on mixed state $\rho: 1 \CPMrightarrow X_1 \times ... \times X_N$, for a measurement scenario $((\hbox{\begin{tikzpicture} [scale=1,transform shape] 

\def\deltax{0.3} 
\def\deltay{0.5} 


\node [dot, fill=\groupStructColour] (mult) at (0,0) {};

\end{tikzpicture}}\!_j^m)_{j=1,...,N})_{m=1,...,M}$. Let $\oplus_{\lambda_j^m \in \Lambda_j^m} G_{\lambda_j^m}$ be the abelian groupoid associated with $\hbox{\begin{tikzpicture} [scale=1,transform shape] 

\def\deltax{0.3} 
\def\deltay{0.5} 


\node [dot, fill=\groupStructColour] (mult) at (0,0) {};

\end{tikzpicture}}\!_j^m$. Let $D \preceq \rho$ be local purity minimum (i.e. $\RelGraph{D}$ has the same nodes of $\RelGraph{\rho}$ and no edge), and consider the empirical model $(\Psi^m)_m$, based on mixed state $D$ instead of $\rho$ and otherwise identical to $(\Phi^m)_m$ (i.e. again for the measurement scenario $((\hbox{\begin{tikzpicture} [scale=1,transform shape] 

\def\deltax{0.3} 
\def\deltay{0.5} 


\node [dot, fill=\groupStructColour] (mult) at (0,0) {};

\end{tikzpicture}}\!_j^m)_{j=1,...,N})_{m=1,...,M}$). 

	Then, since the $\Phi^m$ and $\Psi^m$ are functions are computed by considering scalars in the form $\sigma^\dagger \cdot \rho$ and $\sigma^\dagger \cdot D$ respectively (for the same $\sigma$s), we have that the two empirical models coincide. Thus any empirical model can equivalently be obtained as the empirical model based on some state with discrete associated subgraph. 

	Elaborating a bit further, we conclude that $\Psi^m(\lambda_1^m,...,\lambda_N^m) = 1$ if and only if$\graphNode{(g_1^m,....,g_N^m)} \in \RelGraph{D}$ for some family $(g_j^m \in  G_{\lambda_j^m})_{j=1,...,N}$. The state $D$ can be concretely obtained from $\rho$ by decohering \vspace{-1.5mm} each $X_j$ component in the discrete structure $\hbox{\begin{tikzpicture} [scale=1,transform shape] 

\def\deltax{0.3} 
\def\deltay{0.5} 


\node [dot, fill=\Zaltcolour] (mult) at (0,0) {};

\end{tikzpicture}}\!$ on $X_j$. Furthermore the duplication map $\ZaltcomultSym: X_j \rightarrow X_j \times X_j$ for the discrete structure sends the pure state $\ket{(g_1^m,...,g_j^m,...,g_N^m)}$ to $\ket{(g_1^m,...,g_j^m,g_j^m,...,g_N^m)}$. 

	Now consider the state \vspace{-1mm} $\nu' : 1 \CPMrightarrow \prod_{m=1,...,M} \left(\prod_{j=1,...,N} Z_j^m\right)$ shown in Equation \ref{eqn_LHV}, where $Z_j^m := X_j$. Let $\equiv$ be the equivalence relation such that $(j,m) \equiv (j',m')$ if and only if $\hbox{\begin{tikzpicture} [scale=1,transform shape] 

\def\deltax{0.3} 
\def\deltay{0.5} 


\node [dot, fill=\groupStructColour] (mult) at (0,0) {};

\end{tikzpicture}}\!_j^m$ and $\hbox{\begin{tikzpicture} [scale=1,transform shape] 

\def\deltax{0.3} 
\def\deltay{0.5} 


\node [dot, fill=\groupStructColour] (mult) at (0,0) {};

\end{tikzpicture}}\!_{j'}^{m'}$ are the same classical structure, let $K$ be the set of equivalence classes of $\equiv$ and $q: \{1,...,N\}\times\{1,...,M\} \epim K$ be the associated quotient map (denote the structures by $\hbox{\begin{tikzpicture} [scale=1,transform shape] 

\def\deltax{0.3} 
\def\deltay{0.5} 


\node [dot, fill=\groupStructColour] (mult) at (0,0) {};

\end{tikzpicture}}\!_k$). 

	Then $\nu'$ is obtained from a mixed state $\nu: 1 \CPMrightarrow \prod_{j=1,...,N}\prod_{k \in K} Y_j^k$, where $Y_j^k := X_j$, by \textit{multiplexing} each $Y_j^k$ component in the $\hbox{\begin{tikzpicture} [scale=1,transform shape] 

\def\deltax{0.3} 
\def\deltay{0.5} 


\node [dot, fill=\groupStructColour] (mult) at (0,0) {};

\end{tikzpicture}}\!_k$ structure and then connecting it to all the $Z_j^m$ component wires with $q(j,m) = k$. The state $\nu$, the \textbf{local hidden variable}, is obtained in Equation \ref{eqn_LHV} by multiplexing each $X_j$ component of $D$ in the discrete structure and then connecting it to all the $Y_j^k$ wires with $k \in K$. 

	Tracing out all components of $\nu'$ except for $(Z_j^m)_{j=1,...,N}$ yields a mixed state with discrete associated graph (since all $\hbox{\begin{tikzpicture} [scale=1,transform shape] 

\def\deltax{0.3} 
\def\deltay{0.5} 


\node [dot, fill=\groupStructColour] (mult) at (0,0) {};

\end{tikzpicture}}\!_k$ multiplexings have disappeared) and node set $\{\graphNode{(g_j)_j}\}$. Evaluated against the pure separable state $\rho_{\lambda_1^m} \times ... \times \rho_{\lambda_N^m}$, this yields $1$ if and only if $g_j \in G_{\lambda_j^m}$ for all $j=1,...,N$, i.e. if and only if $\Phi^m(\lambda_1^m,...,\lambda_N^m) = 1$. Therefore $\nu$ is a local hidden variable for the empirical model $(\Phi^m)_m$. 
	\end{proof}

	\begin{proof} (Sub-lemma of Theorem \ref{thm_LHV})
		We have to show that, if $\rho,\sigma$ are two mixed states $1 \CPMrightarrow X$, then the scalar $\sigma^\dagger \cdot \rho$ only depends on the nodes of the subgraphs $\RelGraph{\rho},\RelGraph{\sigma} \leq \completeGraph{X}$. But $\RelGraph{\sigma^\dagger \cdot \rho}$ is non-empty (equivalently, $\sigma^\dagger \cdot \rho=1$) if and only if $\RelGraph{\rho}$ and $\RelGraph{\sigma}$ have some node in common.
	\end{proof}

\section{Conclusions}

	Building on established literature and recent developments in the graph-theoretic characterisation of $\CPMCategory{\fRelCategory}$, we have provided an overview of pure state CQM in the category $\fRelCategory$ of finite sets and relations, and explored mixed state CQM in the same setting. Comparing and contrasting with the category $\fdHilbCategory$ of Hilbert spaces and linear maps, traditional arena of quantum theory, we have presented many solid parallelisms, and many puzzling features. Superposition is well-defined, but the basis of singletons plays a distinguished role; there are many classical structures, but the classical structure associated with singletons is strictly finer than all of the others. Unitarity is a very restrictive property, making pure state quantum mechanics in $\fRelCategory$ not very interesting. 

	However, the existence of many classical structures without enough classical points, and many distinct classical structures for each set of classical points, makes mixed state quantum mechanics in $\fRelCategory$ spicy and exotic. For example, convex mixing of pure (or non-pure) states can result in a pure state. But most importantly,  decohering a state does not in general result in convex mixing of classical points: the interpretation of measurement results as possibilistic mixtures is ill-founded, and  extra care is needed in the treatment of measurements and locality. 

	We gave the full characterisation of decoherence maps, and showed that they are the only demolition measurements needed if the results of measurements are, as traditionally done, tested against classical points. We showed that the degree of purity of a state, characterised in graph theoretic terms as the number of edges present, cannot be measured in our setup, and conclude that $\fRelCategory$ is a local theory. 

	This latter result, however, is based on the assumption that measurements can be tested against classical points to obtain classical data: this is a legitimate assumption in $\fdHilbCategory$, where decoherence produces convex mixing of classical points, but it may turn not be in $\fRelCategory$. Thus we cannot exclude the existence of a more thorough categorical characterisation of the measurement process which would lead to more general empirical models and re-open the question of locality. 

	Furthermore, our definition is restricted to measurements associated to \textbf{commutative} special $\dagger$-Frobenius algebras, and it is interesting to see whether the line of reasoning presented in this work will extend to the non-commutative case (which has already received attention in \cite{BestRel-RelativeFrob}). We hope that these questions will be settled in future work.

\section*{Acknowledgements}
	The author would like to thank Bob Coecke, Dan Marsden, Amar Hadzihasanovic and William Zeng for comments, suggestions and useful discussions, as well as Sukrita Chatterji and Nicol\`o Chiappori for their continued support. Funding from EPSRC and Trinity College is gratefully acknowledged. 

\newpage
\bibliographystyle{eptcs}
\bibliography{bibliography/CategoryTheory,bibliography/CategoricalQM,bibliography/NonLocalityContextuality,bibliography/QuantumComputing,bibliography/ClassicalMechanics,bibliography/LogicComputation,bibliography/Gravitation,bibliography/QFT,bibliography/StatisticalPhysics,bibliography/Misc,bibliography/StefanoGogioso,bibliography/ComputationalLinguistics,bibliography/BestiaryRel}

\newpage

\section{Appendix - Local Hidden Variables}
	
	The framework for non-locality set by \cite{NLC-SheafSeminal} consists of the following basic ingredients:
	\begin{enumerate}
		\item[1.] A commutative semiring $(R,+,0,\times,1)$, generalising the semiring $(\reals^{+},+,0,\times,1)$ in the definition of probabilities (e.g. the semiring of booleans $(\{\bot,\top\},\vee,0,\wedge,1)$, yielding \textit{possibilistic} models).
		
		\item[2.] A set $\mathcal{M}$ of \textbf{measurements}, and for each measurement $\mu \in \mathcal{M}$ a set $O_\mu$ of possible outcomes.
		
		\item[3.] The \textbf{event sheaf} $\mathcal{E}: \mathcal{P}(\mathcal{M})^{op} \longrightarrow \SetCategory$, acting as $U \mapsto \prod_{\mu \in U} O_\mu$ on objects and as $\left(V \subseteq U\right) \mapsto \operatorname{res}^U_V$ on morphisms, where $\operatorname{res}^U_V$ is the restriction map sending $s \in \prod_{\mu \in U} O_\mu$ to $s_{\vert V} \in \prod_{\mu \in V} O_\mu$. The event sheaf sends a subset $U$ of measurements to their set of \textbf{joint outcomes}.
		
		\item[4.] The \textbf{presheaf of distributions} $\mathcal{D}_R\mathcal{E}$, where $\mathcal{D}_R : \SetCategory \longrightarrow \SetCategory$ is the functor defined to send a set $P$ to the set of $R$-distributions on $P$, i.e. those functions $d: P \rightarrow R$ which have finite support and satisfy the normalisation condition $\sum_{s \in P} d(s) = 1$. The presheaf of distributions then sends a set $U \subseteq \mathcal{M}$ of measurements to the set of $R$-distributions over the joint outcomes for the measurements in $U$. The action on morphisms is given by restriction:
		\begin{equation}
			d_{\vert V} := t \mapsto  \!\!\!\! \!\!\!\! \!\!\!\! \!\!\!\! \sum_{s \in \mathcal{E}(U) \text{ with } \operatorname{res}^U_V(s)=t} \!\!\!\! \!\!\!\! \!\!\!\! \!\!\!\! d(s) 
		\end{equation}
		
		\item[5.] The category $\mathcal{P}(\mathcal{M})$ comes with a Grothendieck topology\footnote{In \cite{NLC-SheafSeminal}, $\mathcal{M}$ and all $O_\mu$ are finite, and the Gr. topology on $\mathcal{P}(\mathcal{M})$ is the one corresponding to the discrete topology on $\mathcal{M}$.} making $\mathcal{E}$ into a sheaf. 
		
		\item[6.] A cover $\mathfrak{U}$ for the Grothendieck topology on $\mathcal{P}(\mathcal{M})$, the \textbf{measurement cover}:\footnote{In \cite{NLC-SheafSeminal}, $\mathfrak{U}$ is taken to be any family of subsets with $\cup \mathfrak{U} = \mathcal{M}$ and such that for any $C,C' \in \mathfrak{U}$ we have that $C \subseteq C'$ implies $C=C'$ (i.e. $\mathfrak{U}$ is an antichain).} the sets $C \in \mathfrak{U}$ are the \textbf{measurement contexts}, stipulating all possible sets of \textbf{mutually compatible measurements}.
		
		\item[7.] A compatible family $(e_C)_{C \in \mathfrak{U}}$ for the measurement cover with respect to the presheaf $\mathcal{D}_R\mathcal{E}$, i.e. a family of elements $e_C \in \mathcal{D}_R\mathcal{E}(C)$ such that $(e_C)_{\vert C \cap C'} = (e_{C'})_{\vert C \cap C'}$ for all $C,C' \in \mathfrak{U}$. This will be called a \textbf{(no-signalling) empirical model}, assigning an $R$-distribution to the joint outcomes of all sets of mutually compatible measurements.

	\end{enumerate}

	The ingredients above form a \textbf{measurement scenario}. A \textbf{global section} for an empirical model $(e_C)_{C \in \mathcal{M}}$, which we shall also refer to as a \textbf{local hidden variable}, is an $R$-distribution $d \in \mathcal{D}_R\mathcal{E}(\mathcal{M})$ such that $e_C = d_{\vert C}$ for all $C \in \mathfrak{U}$. An empirical model is \textbf{local} if it admits a global section, i.e. if there is a $R$-distribution $d$ on the joint outcomes of all measurements which, when restricted to the measurement context specified by $\mathfrak{U}$, yields the same $R$-distributions as the empirical model; an empirical model is \textbf{non-local} if no such $R$-distribution $d$ exists.\\

	\noindent This framework can be given the following (approximate) semantics in the context of CQM:

	\begin{enumerate}
		\item[1.] We consider a compact-closed\footnote{Strictly not necessary, as the CPM construction, the only reason the requirement of compact-closedness is there, can be replaced with the CP construction from \cite{CQM-PicturesCompletePositivity}.} $\dagger$-SMC $\CategoryC$ enriched over monoids, and let $(R,+,0,\tensor,1)$ be its semiring of scalars.

		\item[2.] We consider measurements to be certain causal morphisms $\mu : X_\mu \CPMrightarrow X_\mu \tensor Z_\mu$ in $\CPMCategory{\CategoryC}$, each satisfying idempotence and self-adjointness with respect to some classical structure $\hbox{\begin{tikzpicture} [scale=1,transform shape] 

\def\deltax{0.3} 
\def\deltay{0.5} 


\node [dot, fill=\classicalStructColour] (mult) at (0,0) {};

\end{tikzpicture}}\!_\mu$ on $Z_\mu$ (see \cite{CQM-QuantumMeasuNoSums} and the upcoming \cite{CQM-QCSnotes} for more details). We take some finite set $\mathcal{M}$ of them. For each measurement, we take the classical points $(G^\mu_\lambda)_{\lambda \in \Lambda_\mu}$ of the associated classical structure $\hbox{\begin{tikzpicture} [scale=1,transform shape] 

\def\deltax{0.3} 
\def\deltay{0.5} 


\node [dot, fill=\classicalStructColour] (mult) at (0,0) {};

\end{tikzpicture}}\!_\mu$ as the possible outcomes, assuming that each classical structure involved has a finite set of classical points.

		\item[3.] The event sheaf is defined as sending a set $U \subseteq \mathcal{M}$ of measurements to the following set:
		\begin{equation}
			\mathcal{E}(U) := \suchthat{\bigtensor_{\mu \in U} G^\mu_{\lambda_\mu} }{(\lambda_\mu)_\mu \in \prod_{\mu \in U} \Lambda_\mu}
		\end{equation}

		\item[4.] The presheaf of distributions is defined as sending a set $U \subseteq \mathcal{M}$ of measurements to the set of all possible states $d: \tensorUnit \CPMrightarrow \tensor_{\mu \in U} Z_\mu$ satisfying the following normalisation condition:
		\begin{equation}\label{eqn_NormalisationConditionDistributionsSMC}
			\left(\bigtensor_{\mu \in U} \cap_{Z_\mu} \right) \cdot d = 1
		\end{equation}
		 i.e. such that $d$ is causal. We see one such state $d$ as a $R$-distribution on $\mathcal{E}(U)$ by defining:
		\begin{equation}
			d(\bigtensor_{\mu \in U} G^\mu_{\lambda_\mu}) := \left(\bigtensor_{\mu \in U} \rho^\dagger_{G^\mu_{\lambda_\mu}}\right) \cdot d \in R
		\end{equation}
		where $\rho_{G^\mu_{\lambda_\mu}}$ is the pure state in $\CPMCategory{\CategoryC}$ corresponding to pure state $\ket{G^\mu_{\lambda_\mu}}$ from $\CategoryC$. If for each $\mu \in \mathcal{M}$ we have that $\id{} = \sum_{\lambda \in \Lambda_\mu} \ket{G^\mu_\lambda}\bra{G^\mu_\lambda}$, then the normalisation condition of Equation \ref{eqn_NormalisationConditionDistributionsSMC} can equivalently be written as the more familiar:
		\begin{equation}
			\sum_{s \in \mathcal{E}(U)} d(s) = 1
		\end{equation}
		The restriction $d_{\vert V}$ of a distribution $d: \tensorUnit \CPMrightarrow \tensor_{\mu \in U} Z_\mu$ for $V \subseteq U$ is given by:
		\begin{equation}\label{eqn_RestrictionOperationDistributionsSMC}
			d_{\vert V} := \left( \bigtensor_{\mu \in U} \xi_\mu \right) \cdot d \text{, where } \xi_\mu = 
				\begin{cases}
					\id{Z_\mu} \text{ if } \mu \in V\\
					\cap_{Z_\mu} \text{ if } \mu \notin V  
				\end{cases}
		\end{equation}

		\item[5.] The topology on $\mathcal{P}(\mathcal{M})$ is as before.

		\item[6.] A measurement cover $\mathfrak{U}$ is chosen as before.

		\item[7.] A compatible family is chosen to be a family of states $e_C: \tensorUnit \CPMrightarrow \tensor_{\mu \in C} Z_\mu$ for all $C \in \mathcal{M}$ satisfying the normalisation condition of Equation \ref{eqn_NormalisationConditionDistributionsSMC} and such that $(e_C)_{\vert C \cap C'} = (e_{C'})_{\vert C \cap C'}$ for all $C,C' \in \mathfrak{U}$, under the restriction operation defined in Equation \ref{eqn_RestrictionOperationDistributionsSMC}.
	\end{enumerate}

	\noindent We focus our attention to a particular kind of empirical models, based on measurements of some mixed state $\rho$. Let $\mathfrak{U}$ be a measurement cover, and name its distinct measurement contexts $C_1,...,C_M$. Let each measurement context $C_m$ consist of $N$ measurements in the form:
	\begin{equation}	
		\mu^m_j : X_j \CPMrightarrow X_j \tensor Z^m_j \text{, for } j=1,...,N
	\end{equation}
	and let $\hbox{\begin{tikzpicture} [scale=1,transform shape] 

\def\deltax{0.3} 
\def\deltay{0.5} 


\node [dot, fill=\classicalStructColour] (mult) at (0,0) {};

\end{tikzpicture}}\!_j^m$ be the respective classical structures. Fix some causal state:
	\begin{equation}
		\rho: \tensorUnit \CPMrightarrow \bigtensor_{i=1,...,N} X_i
	\end{equation}
	Then an empirical model $e_m : \tensorUnit \CPMrightarrow \tensor_{j=1,...,N} Z^m_j$ can be defined as follows, for $m=1,...,M$:
	\begin{equation}
		e_{C_m} := \left( \bigtensor_{j=1}^N \left( \cap_{X_j} \tensor \id{Z_j^m} \right)\right) \cdot \rho 
	\end{equation}

	\noindent Local hidden variables are defined as distributions on the \textit{set} of measurements, rather than on the family $\mu^m_j$ indexed by the measurement contexts. We define an equivalence relation by setting $(j,m) \equiv (j',m')$ if and only if $\mu^m_j = \mu^{m'}_{j'}$: we index the equivalence classes by $\kappa =1,...,K$ and let $q : \{1,...,N\}\times\{1,...,M\} \epim \{1,...,K\}$ be the quotient map. Then $\mu^m_j = \mu^{m'}_{j'}$ if and only if $q(j,m) = q(j',m')$: we define $Y_j^\kappa := Z_j^m$ and $\hbox{\begin{tikzpicture} [scale=1,transform shape] 

\def\deltax{0.3} 
\def\deltay{0.5} 


\node [dot, fill=\classicalStructColour] (mult) at (0,0) {};

\end{tikzpicture}}\!_\kappa := \hbox{\begin{tikzpicture} [scale=1,transform shape] 

\def\deltax{0.3} 
\def\deltay{0.5} 


\node [dot, fill=\classicalStructColour] (mult) at (0,0) {};

\end{tikzpicture}}\!_j^m$ for some (any) $m$ such that $q(j,m) = \kappa$, where $\kappa = 1,...,K$. The \textbf{local map} for the measurement scenarion is then defined to be the map $\left(\tensor_i \tensor_\kappa Y_i^\kappa\right) \CPMrightarrow \left(\tensor_m \tensor_j Z_j^m\right)$ given as follows (see figure \ref{eqn_LocalMap} for a graphical definition):
	\begin{enumerate}
		\item[a.] we group the input wires in $N$ groups of $K$ wires
		\item[b.] we group the output wires in $M$ groups of $N$ wires
		\item[c.] each $Y_i^\kappa$ input wire is connected to a $\hbox{\begin{tikzpicture} [scale=1,transform shape] 

\def\deltax{0.3} 
\def\deltay{0.5} 


\node [dot, fill=\classicalStructColour] (mult) at (0,0) {};

\end{tikzpicture}}\!_\kappa$ node
		\item[d.] for all $i,j$ and $\kappa,m$, the $\hbox{\begin{tikzpicture} [scale=1,transform shape] 

\def\deltax{0.3} 
\def\deltay{0.5} 


\node [dot, fill=\classicalStructColour] (mult) at (0,0) {};

\end{tikzpicture}}\!_\kappa$ node of each $Y_i^\kappa$ input wire is connected to the $Z_j^m$ output wire if and only if $i=j$ and $q(j,m) = \kappa$
	\end{enumerate}

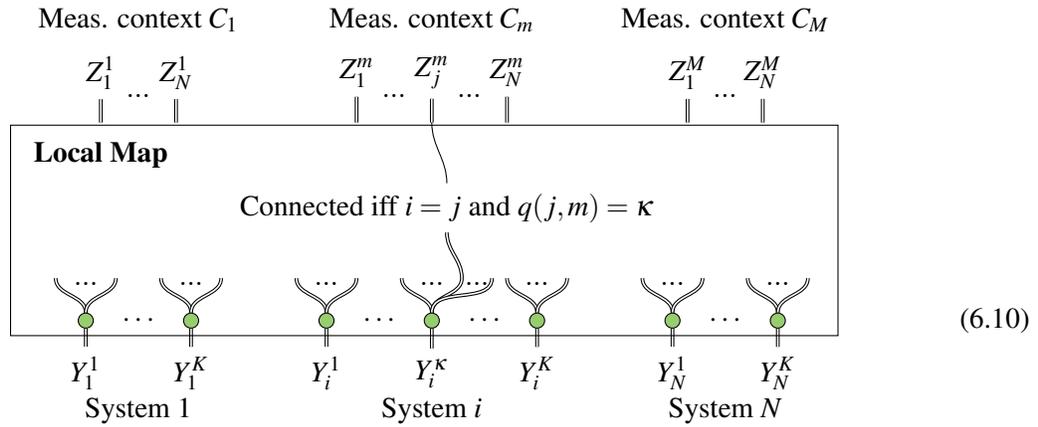
	\begin{equation}\label{eqn_LocalMap}
		\hbox{\begin{tikzpicture}[node distance = 10mm]



\node (a1sys1dot) [smalldot, fill = \Zcolour] {};
\node (a1sys1dotl) [above of = a1sys1dot, xshift = -4mm,yshift = -3mm] {};
\node (a1sys1dotr) [above of = a1sys1dot, xshift = +4mm,yshift = -3mm] {};
\node (a1sys1dotellipsis) [above of = a1sys1dot, yshift = -2mm,yshift = -3mm] {$...$};

\node (a1sys1) [below of = a1sys1dot,yshift = +3mm] {$Y_1^1$};

\node (sys1dotellipsis) [right of = a1sys1dot,xshift = -3mm] {$\cdot \cdot \cdot$};
\node (sys1label) [below of = sys1dotellipsis, yshift = -2mm] {System $1$};

\node (aMsys1dot) [smalldot, right of = sys1dotellipsis,xshift = -3mm, fill = \Zcolour] {};
\node (aMsys1dotl) [above of = aMsys1dot, xshift = -4mm,yshift = -3mm] {};
\node (aMsys1dotr) [above of = aMsys1dot, xshift = +4mm,yshift = -3mm] {};
\node (aMsys1dotellipsis) [above of = aMsys1dot, yshift = -2mm,yshift = -3mm] {$...$};

\node (aMsys1) [below of = aMsys1dot,yshift = +3mm] {$Y_1^K$};

\node (sys1sysiellipsis) [right of = aMsys1dot, yshift = -3mm, xshift = -1mm] {};


\node (a1sysidot) [smalldot, fill = \Zcolour, right of = sys1sysiellipsis, yshift = +3mm, xshift = -1mm] {};
\node (a1sysidotl) [above of = a1sysidot, xshift = -4mm,yshift = -3mm] {};
\node (a1sysidotr) [above of = a1sysidot, xshift = +4mm,yshift = -3mm] {};
\node (a1sysidotellipsis) [above of = a1sysidot, yshift = -2mm,yshift = -3mm] {$...$};

\node (a1sysi) [below of = a1sysidot,yshift = +3mm] {$Y_i^1$};

\node (sysidotellipsis) [right of = a1sysidot,xshift = -3mm] {$\cdot \cdot \cdot$};

\node (arsysidot) [smalldot, fill = \Zcolour, right of = sysidotellipsis,xshift = -3mm] {};
\node (arsysidotl) [above of = arsysidot, xshift = -4mm,yshift = -3mm] {};
\node (arsysidotc) [above of = arsysidot, xshift = +4mm,yshift = -3mm] {};
\node (arsysidotr) [above of = arsysidot, xshift = +8mm,yshift = -3mm] {};
\node (arsysidotellipsis) [above of = arsysidot, yshift = -2mm,yshift = -3mm] {$...$};
\node (arsysidotellipsis) [above of = arsysidot, xshift = +6mm,yshift = -2mm,yshift = -3mm] {$...$};

\node (arsysi) [below of = arsysidot,yshift = +3mm] {$Y_i^\kappa$};

\node (sysilabel) [below of = arsysidot, yshift = -2mm] {System $i$};
\node (sysidotellipsis2) [right of = arsysidot,xshift = -3mm] {$\cdot \cdot \cdot$};

\node (aMsysidot) [smalldot, fill = \Zcolour, right of = sysidotellipsis2,xshift = -3mm] {};
\node (aMsysidotl) [above of = aMsysidot, xshift = -4mm,yshift = -3mm] {};
\node (aMsysidotr) [above of = aMsysidot, xshift = +4mm,yshift = -3mm] {};
\node (aMsysidotellipsis) [above of = aMsysidot, yshift = -2mm,yshift = -3mm] {$...$};

\node (aMsysi) [below of = aMsysidot,yshift = +3mm] {$Y_i^K$};

\node (sysisysNellipsis) [right of = aMsysidot, yshift = -3mm, xshift = -1mm] {};


\node (a1sysNdot) [smalldot, fill = \Zcolour, right of = sysisysNellipsis, yshift = +3mm, xshift = -1mm] {};
\node (a1sysNdotl) [above of = a1sysNdot, xshift = -4mm,yshift = -3mm] {};
\node (a1sysNdotr) [above of = a1sysNdot, xshift = +4mm,yshift = -3mm] {};
\node (a1sysNdotellipsis) [above of = a1sysNdot, yshift = -2mm,yshift = -3mm] {$...$};

\node (a1sysN) [below of = a1sysNdot,yshift = +3mm] {$Y_N^1$};

\node (sysNdotellipsis) [right of = a1sysNdot,xshift = -3mm] {$\cdot \cdot \cdot$};
\node (sysNlabel) [below of = sysNdotellipsis, yshift = -2mm] {System $N$};

\node (aMsysNdot) [smalldot, fill = \Zcolour, right of = sysNdotellipsis,xshift = -3mm] {};
\node (aMsysNdotl) [above of = aMsysNdot, xshift = -4mm,yshift = -3mm] {};
\node (aMsysNdotr) [above of = aMsysNdot, xshift = +4mm,yshift = -3mm] {};
\node (aMsysNdotellipsis) [above of = aMsysNdot, yshift = -2mm,yshift = -3mm] {$...$};

\node (aMsysN) [below of = aMsysNdot,yshift = +3mm] {$Y_N^K$};

\begin{pgfonlayer}{background}
\draw[double,out=90,in=270] (a1sys1) to (a1sys1dot);
\draw[double,out=90,in=270] (a1sys1dot) to (a1sys1dotl);
\draw[double,out=90,in=270] (a1sys1dot) to (a1sys1dotr);

\draw[double,out=90,in=270] (aMsys1) to (aMsys1dot);
\draw[double,out=90,in=270] (aMsys1dot) to (aMsys1dotl);
\draw[double,out=90,in=270] (aMsys1dot) to (aMsys1dotr);

\draw[double,out=90,in=270] (a1sysi) to (a1sysidot);
\draw[double,out=90,in=270] (a1sysidot) to (a1sysidotl);
\draw[double,out=90,in=270] (a1sysidot) to (a1sysidotr);

\draw[double,out=90,in=270] (arsysi) to (arsysidot);
\draw[double,out=90,in=270] (arsysidot) to (arsysidotl);
\draw[double,out=90,in=270] (arsysidot) to (arsysidotc);
\draw[double,out=90,in=270] (arsysidot) to (arsysidotr);

\draw[double,out=90,in=270] (aMsysi) to (aMsysidot);
\draw[double,out=90,in=270] (aMsysidot) to (aMsysidotl);
\draw[double,out=90,in=270] (aMsysidot) to (aMsysidotr);

\draw[double,out=90,in=270] (a1sysN) to (a1sysNdot);
\draw[double,out=90,in=270] (a1sysNdot) to (a1sysNdotl);
\draw[double,out=90,in=270] (a1sysNdot) to (a1sysNdotr);

\draw[double,out=90,in=270] (aMsysN) to (aMsysNdot);
\draw[double,out=90,in=270] (aMsysNdot) to (aMsysNdotl);
\draw[double,out=90,in=270] (aMsysNdot) to (aMsysNdotr);

\end{pgfonlayer}



\node (alfa1dotellipsis) [above of = sys1dotellipsis, yshift = 20mm] {$...$};
\node (alfa1label) [above of = alfa1dotellipsis, yshift = 0mm] {Meas. context $C_1$};

\node (alfa1sys1anchor) [left of = alfa1dotellipsis, yshift = -5mm, xshift = 5mm] {};
\node (alfa1sys1) [above of = alfa1sys1anchor,yshift = -2mm] {$Z_1^1$};

\node (alfa1sysNanchor) [right of = alfa1dotellipsis,yshift = -5mm, xshift = -5mm] {};
\node (alfa1sysN) [above of = alfa1sysNanchor,yshift = -2mm] {$Z_N^1$};


\node (alfasanchor) [above of = arsysidot, yshift = 20mm] {};
\node (alfaslabel) [above of = alfasanchor, yshift = 0mm] {Meas. context $C_m$};
\node (alfassysianchor) [below of = alfasanchor, yshift = 5mm] {};
\node (alfassysi) [above of = alfassysianchor,yshift = -2mm] {$Z_j^m$};

\node (alfasdotellipsisl) [left of = alfasanchor, xshift = 5mm] {$...$};
\node (alfasdotellipsisr) [right of = alfasanchor, xshift = -5mm] {$...$};

\node (alfassys1anchor) [left of = alfasdotellipsisl, yshift = -5mm, xshift = 5mm] {};
\node (alfassys1) [above of = alfassys1anchor,yshift = -2mm] {$Z_1^m$};

\node (alfassysNanchor) [right of = alfasdotellipsisr,yshift = -5mm, xshift = -5mm] {};
\node (alfassysN) [above of = alfassysNanchor,yshift = -2mm] {$Z_N^m$};


\node (alfaSdotellipsis) [above of = sysNdotellipsis, yshift = 20mm] {$...$};
\node (alfaSlabel) [above of = alfaSdotellipsis, yshift = 0mm] {Meas. context $C_M$};

\node (alfaSsys1anchor) [left of = alfaSdotellipsis, yshift = -5mm, xshift = 5mm] {};
\node (alfaSsys1) [above of = alfaSsys1anchor,yshift = -2mm] {$Z^M_1$};

\node (alfaSsysNanchor) [right of = alfaSdotellipsis,yshift = -5mm, xshift = -5mm] {};
\node (alfaSsysN) [above of = alfaSsysNanchor,yshift = -2mm] {$Z^M_N$};

\begin{pgfonlayer}{background}
\draw[double,out=90,in=270] (alfa1sys1anchor) to (alfa1sys1);
\draw[double,out=90,in=270] (alfa1sysNanchor) to (alfa1sysN);

\draw[double,out=90,in=270] (alfassys1anchor) to (alfassys1);
\draw[double,out=90,in=270] (alfassysianchor) to (alfassysi);
\draw[double,out=90,in=270] (alfassysNanchor) to (alfassysN);

\draw[double,out=90,in=270] (alfaSsys1anchor) to (alfaSsys1);
\draw[double,out=90,in=270] (alfaSsysNanchor) to (alfaSsysN);
\end{pgfonlayer}


\node (label) [above of = arsysidotc, xshift = -2mm,yshift = -2mm] {Connected iff $i=j$ and $q(j,m)=\kappa$};

\node (localmaplabel) [below of = alfa1sys1anchor, yshift = 7mm] {\textbf{Local Map}};

\begin{pgfonlayer}{background}
\node (box) [box, above of = label, yshift = -13mm, xshift = -3mm, minimum width = 110mm, minimum height = 28mm, fill = none] {};
\end{pgfonlayer}

\begin{pgfonlayer}{background}
\draw[double,out=90,in=270] (arsysidotc.270) to (label);
\draw[-,out=90,in=270] (label) to (alfassysianchor.90);
\end{pgfonlayer}

\end{tikzpicture}}
	\end{equation}

	In this context, we define a local hidden variable to be some state $\nu : \tensorUnit \CPMrightarrow \tensor_i \tensor_\kappa Y_i^\kappa$ such that:
	\begin{equation}
		\left(\left(\bigtensor_{m=1}^{r-1}\bigtensor_{j=1}^{N} \cap_{Z_j^m} \right) \tensor 
		\left(\bigtensor_{j=1}^{N} \id{Z_j^r} \right) \tensor
		\left(\bigtensor_{m=r+1}^{M}\bigtensor_{j=1}^{N} \cap_{Z_j^m} \right) \right) \cdot \operatorname{localmap} \cdot \nu = e_{C_r}
	\end{equation}
	for all $r=1,...,M$. In terms of $R$-distributions, this is the same as the definition from \cite{NLC-SheafSeminal}.

\end{document}